\documentclass[graybox]{svmult}

\usepackage{mathptmx}       % selects Times Roman as basic font
\usepackage{helvet}         % selects Helvetica as sans-serif font
\usepackage{courier}        % selects Courier as typewriter font
\usepackage{type1cm}        % activate if the above 3 fonts are
\usepackage{makeidx}         % allows index generation
\usepackage{graphicx}        % standard LaTeX graphics tool
\usepackage{multicol}        % used for the two-column index
\usepackage[bottom]{footmisc}% places footnotes at page bottom

\makeindex             % used for the subject index
\usepackage[square,authoryear]{natbib}
\usepackage{framed}
\usepackage{latexsym,amsmath,amscd,amssymb,graphics}
\usepackage{enumerate,framed,comment}
\usepackage{color}
\usepackage{verbatim}
\usepackage[margin=2.5cm]{geometry}
\usepackage{epsfig,epstopdf,color,enumerate}% Include figure files
\usepackage{caption}
\usepackage{graphicx, subfigure, psfrag} % subfigure package
\usepackage{placeins}
\usepackage{sidecap}

 \def\Xint#1{\mathchoice
{\XXint\displaystyle\textstyle{#1}}%
{\XXint\textstyle\scriptstyle{#1}}%
{\XXint\scriptstyle\scriptscriptstyle{#1}}%
{\XXint\scriptscriptstyle\scriptscriptstyle{#1}}%
\!\!\int}
\def\XXint#1#2#3{{\setbox0=\hbox{$#1{#2#3}{\int}$ }
\vcenter{\hbox{$#2#3$ }}\kern-.5\wd0}}

\def\dashint{\Xint-}

\def\pint{\frac{1}{2\pi}\int_{-\pi}^\pi\! {\rm d}}

\def\dpint{\frac{1}{2\pi}\dashint_{-\pi}^\pi \! {\rm d}}

\def\R{\mathbb{R}}

\newcommand{\phm}[1]{\phantom{#1}}

\def\bq{\begin{equation}}
\def\eq{\end{equation}}
\def\bqy{\begin{eqnarray}}
\def\eqy{\end{eqnarray}}

\def\p{\partial}

\begin{document}

\title*{Gradient flows in the normal and K\"ahler metrics and triple bracket generated metriplectic systems}
 \titlerunning{Gradient flows and metriplectic systems}
\author{Anthony M. Bloch, Philip J. Morrison, and Tudor S. Ratiu}
\institute{Anthony M. Bloch \at Department of Mathematics, The University of Michigan, 530 Church Street,  
Ann Arbor, MI 48109-1043, USA. 
\texttt{abloch@umich.edu}
\and Philip J. Morrison \at Department of Physics and Institute for Fusion Studies, University of Texas,
2515 Speedway Stop C1600, Austin, TX 78712-0264, USA.  
\texttt{morrison@physics.utexas.edu}
\and Tudor S. Ratiu \at Department of Mathematics and Bernoulli Center, 
Ecole Polytechnique F\'ed\'erale de Lausanne, CH-1015 Lausanne, Switzerland. 
\texttt{tudor.ratiu@epfl.ch} }

\maketitle

\abstract{
The dynamics of gradient and Hamiltonian flows with particular application to 
flows on adjoint orbits of a Lie group and the extension of this setting to flows on a loop group are discussed.
Different types of gradient flows that arise from different metrics including the  so-called
normal metric on adjoint orbits of a Lie group and the K\"ahler metric are compared.
It is discussed how a  K\"ahler metric can arise from a  complex structure induced
by  the Hilbert transform.  Hybrid and metriplectic flows that arise when 
one has both Hamiltonian and gradient components are examined. A class of metriplectic systems that is generated by completely antisymmetric triple brackets is described and for finite-dimensional systems given a Lie algebraic interpretation.  A variety of explicit examples of the several types of flows are given. 
}


%\tableofcontents

%%%%%%%%%%%%%%%%%%%
%%%%%%%%%%%%%%%%%%%
%%%%%%%%%%%%%%%%%%%
\section{Introduction}
\label{Intro}

Dynamical systems, finite or infinite,  that describe physical phenomena typically have parts that are in some 
sense Hamiltonian and parts that can be recognized as dissipative, with the Hamiltonian part being generated by a Poisson bracket and the dissipative part being some kind of gradient flow.   The  description of Hamiltonian systems has received much attention over nearly two centuries and, although some forms of dissipation have  received general attention, the understanding and classification of  dissipative dynamics is a much broader topic and consequently less well developed.   Early modern treatments of geometric Hamiltonian mechanics include those of   \cite{souriau} and \cite{AbMa1978},  and the literature on this topic is now immense.  A special type of gradient flow that preserves invariants,  the double bracket formalism  described in \cite{Brockett1991} (see, e.g., \cite{Bloch1990}, \cite{bloch03}),  is a formalism that occurs in a variety of contexts   (see \cite{BlKrMaRa1994,BlKrMaRa1996}) and is well-adapted to practical numerical computations (see \cite{VaCarYo,FlierlMor}).   Examples of  infinite-dimensional gradient flows  include the Cahn-Hilliard systems  (see  \cite{otto}) and the celebrated  Ricci flows (see \cite{hamilton,chow}),  which are nonlinear diffusion-like equations.  A general form for combined  Hamiltonian and gradient flows was described  in \cite{Morrison1986},  where such flows were termed {\it metriplectic} flows (see also \cite{Oettinger2006,Morrison2009,LiMi2012}).  Thus, it is evident that there are a variety Hamiltonian and dissipative flows, and the purpose of this paper is to explore the form and geometric structure of such flows in both the ode and pde contexts.

Specifically,  in this paper we discuss  the dynamics of gradient and Hamiltonian 
flows,  with particular application to flows on adjoint orbits
of a Lie group and the extension of this setting to flows on a loop group.
We compare the different types of gradient 
flows that arise from different metrics, in particular,  the so-called
normal metric on adjoint orbits of a Lie group  and the K\"ahler metric.
We discuss how a K\"ahler metric can arise from the complex structure induced
from the Hilbert transform. We also consider flows that arise when 
one has both Hamiltonian and gradient structures present. In particular, 
we discuss metriplectic flows, flows that produce entropy while conserving energy.  We consider such flows  in both the finite and infinite settings, and discuss  a general class of metriplectic flows that arise from completely antisymmetric triple brackets.  For finite systems,  we show how the triple bracket has a natural Lie algebraic  formulation, and for infinite systems we give a procedure for constructing a quite general class of metriplectic pdes.  We also consider, hybrid flows, of Hamiltonian and gradient form, that dissipate energy.  Several examples of hybrid and metriplectic flows are given, including finite systems such as the Toda lattice on $\R$ and metriplectic $\mathfrak{so}(3)$ brackets.   Various infinite-dimensional examples including a $1+1$ dissipative systems that conserves energy, and hybrid systems such as  the KdV with dissipation, the  \cite{OS1969} equation that describes Landau damping, and others.

The paper is organized as follows.  In section \ref{metrics_on_orbits} we review material need for latter development.   In particular, we discuss  metrics on adjoint  orbits, Toda flows and the double bracket formulation.   Sections \ref{gradients_section} and \ref{sec_metriplectic} contain the main new results of the paper as described above.  In section \ref{gradients_section} we discuss metrics on loop groups and  related gradient flows, while in section \ref{sec_metriplectic} we discuss our results on  metriplectic systems, in both  finite- and infinite-dimensions, and give  examples.

%%%%%%%%%%%%%%%%%%%
%%%%%%%%%%%%%%%%%%%
%%%%%%%%%%%%%%%%%%%
\section{Metrics on adjoint orbits of compact Lie groups
and associated dynamical systems}
\label{metrics_on_orbits}

%%%%%%%%%%%%%%%%%%%
%%%%%%%%%%%%%%%%%%%
\subsection{Double bracket systems} 
\label{double}

Let ${\mathfrak g}_u$ be the compact real form of a complex semisimple 
Lie algebra ${\mathfrak g}$, $G_u$ a compact connected real Lie group
with Lie algebra $\mathfrak{g}_u$, and $\kappa$ the Killing form (on
$\mathfrak{g}$ or $\mathfrak{g}_u$, depending on the context).

The ``normal'' metric on the adjoint orbit $\mathcal{O}$ of $G_u$ 
through $L_0 \in \mathfrak{g}_u$ (see \cite{Atiyah1982}, 
\cite[Chapter 8]{Besse2008}) is given as follows. Decompose orthogonally 
${\mathfrak g}_u = {\mathfrak g}_u^L
\oplus{\mathfrak g}_{uL}$, relative to to the invariant inner 
product $\left\langle~~,~~\right\rangle: = -\kappa (~~,~~)$,  
where ${\mathfrak g}_{uL}: = \ker \operatorname{ad}_L$ is the centralizer 
of $L$ and ${\mathfrak g}_u^L= \operatorname{range} \operatorname{ad}_L$; 
as usual, $\operatorname{ad}_L:= [L, \cdot ]$.  For $X\in{\mathfrak g}_u$ 
denote by $X^L\in{\mathfrak g}_u^L$ and $X_L \in {\mathfrak g}_{uL}$ 
the orthogonal projections of $X$ on ${\mathfrak g}_u^L$ and 
${\mathfrak g}_{uL}$, respectively. Recall that a general vector tangent
at $L$ to the adjoint orbit $\mathcal{O}$ is 
necessarily of the form $[L, X]$ for some $X \in \mathfrak{g}_u$. The
\textit{normal metric} on $\mathcal{O}$ is the $G_u$-invariant 
Riemannian metric given 
by \begin{equation}
\label{normal_metric}
\left\langle[L,X], [L,Y] \right\rangle_{\rm normal} : = 
\left\langle X^L, Y^L \right\rangle\,
\end{equation}
for any $X,Y \in \mathfrak{g}_u$.

Fix $N \in \mathfrak{g}_{u}$ and consider the flow on the adjoint 
orbit $\mathcal{O}$ of $G_u$ through 
$L_0 \in \mathfrak{g}_u$ given by 
\begin{equation}
\frac{d}{dt}L(t) = \lbrack L(t),\lbrack L(t),N\rbrack\rbrack\,,
\qquad L(0) = L_0 \in \mathfrak{g}_u\,. 
\label{dbflow}
\end{equation}
We recall the following well-known result 
(\cite{Brockett1991}, \cite{Brockett1994}, \cite{BlBrRa1990}, \cite{BlBrRa1992}, \cite{BlFlRa1990}, \cite{BlIs2005}).

\begin{proposition}  
\label{double_bracket_prop}
The vector field given by the ordinary differential
equation \eqref{dbflow} is the gradient of the function 
$H(L)=\kappa (L,N)$ relative to the normal metric on $\mathcal{O}$.
\end{proposition}

\begin{proof}  By the definition of the gradient 
$\operatorname{grad}H(L) \in T_L\mathcal{O}\subset 
\mathfrak{g}_{u}$ relative to the normal metric, we have 
for any 
$L \in \mathcal{O}$ and $\delta L \in \mathfrak{g}_u$,
\begin{equation}
dH(L)\cdot\lbrack L,\delta L\rbrack = 
\langle\operatorname{grad}\,H(L),\lbrack L,\delta L\rbrack
\rangle_{\rm normal} \label{grad}
\end{equation}
where $\cdot$ denotes the natural pairing between 1-forms and tangent 
vectors and $\lbrack L,\delta L\rbrack$ is an arbitrary tangent vector 
at $L$ to $\mathcal{O}$.  Set $\text {grad}\,H(L) = 
\lbrack L,X\rbrack
=\lbrack L,X^L\rbrack$.  Then (\ref{grad}) becomes
$$
-\langle\lbrack L,\delta L\rbrack ,N\rangle = 
\langle\lbrack L,X\rbrack ,\lbrack L,\delta L\rbrack 
\rangle_{\rm normal}
$$
or, equivalently,
$$
\langle\lbrack L,N\rbrack ,\delta L\rangle = 
\langle X^L,\delta L^L\rangle = 
\langle X^L,\delta L\rangle\,.
$$
Since $[L,N] \in \mathfrak{g}_u^L$, this implies that 
$X^L=\lbrack L,N\rbrack$, and hence 
$\operatorname{grad}\,H(L) = \lbrack L,\lbrack L,N\rbrack\rbrack$,
as stated. \hfill $\blacksquare$
\end{proof}

The same computation, for a general function 
$H \in C ^{\infty}(\mathfrak{g}_{u})$, yields
\begin{equation}
\label{gen_double_bracket}
\operatorname{grad}\,H(L) = - [L,[L, \nabla H(L)]]
\end{equation}
where $\nabla H(L)$ denotes the gradient of the function $H$
relative to the invariant inner product 
$\left\langle~~,~~\right\rangle: = -\kappa (~~,~~)$, i.e.,
$dH(L)\cdot X = \left\langle \nabla H(L), X \right\rangle$ for
any $X \in \mathfrak{g}_u$.

%%%%%%%%%%%%%%%%%%%
%%%%%%%%%%%%%%%%%%%
\subsection{The finite Toda system}

 The double bracket equation \eqref{dbflow} is intimately related
to the finite non-compact Toda lattice system. This is a Hamiltonian system
modeling $n$ particles moving freely on the $x$-axis
and interacting under an exponential potential. Denoting the position
of the $k$th particle by $x_k$, the Hamiltonian is given by
\[
H(x,y)={1\over 2}\sum_{k=1}^{n}y^2_k
+\sum_{k=1}^{n-1}e^{x_k-x_{k+1}}
\]
and hence the associated Hamiltonian equations are
%-----------------------------
\begin{align}
\dot x_k =   {\partial H\over \partial y_k}=y_k  \,,  \qquad
\dot y_k =  -{\partial H\over \partial x_k}= 
e^{x_{k-1}-x_k}-e^{x_{k}-x_{k+1}} \,,
\label{todaeqns}
\end{align}
%-----------------------------
where we use the conventions $e^{x_0-x_1}=e^{x_n-x_{n+1}}=0$,
which corresponds to formally setting $x_0=-\infty$ and
$x_{n+1}=+\infty$.

This system of equations has an extraordinarily rich structure. Part
of this is revealed by Flaschka's change of
variables (\cite{Flaschka1974}) given by
%-----------------------------
\begin{equation}
a_k=\frac{1}{2}e^{(x_k-x_{k+1})/2} \quad \mbox{and} 
\quad b_k=-\frac{1}{2}y_k\,.
\end{equation}
%-----------------------------
which transform \eqref{todaeqns} to
%-----------------------------
\[
\left\{
\begin{aligned}
\dot a_k&= a_k(b_{k+1}-b_k)    \,, \quad k=1,\dots,n-1  \,,  \\
\dot b_k&= 2(a_k^2 -a_{k-1}^2) \,, \quad k=1,\dots,n    \,,
\end{aligned}
\right.
\]
%-----------------------------
with the boundary conditions $a_0=a_n=0$. This system is equivalent
to the Lax equation
%-----------------------------
\begin{equation}
{d\over dt}L=[B,L]=BL-LB\,,
\label{lax}
\end{equation}  where
%---------------------------------------
\begin{equation}
\label{lb_matrices}
L  =
\left(
\begin{matrix}
b_{1}  & a_{1} & 0      & \cdots  & 0       \\
a_{1}  & b_{2} & a_{2}  & \cdots  & 0       \\
\vdots &       & \ddots &         & \vdots  \\
     0 &       & \cdots & b_{n-1} & a_{n-1} \\
     0 &       & \cdots & a_{n-1} & b_{n}
\end{matrix}
\right)
\,,
%-----------------------------
\qquad  B   =
\left(
\begin{matrix}
   0 & a_1 & 0      & \cdots  & 0               \\
-a_1 & 0   & a_2    & \cdots  & 0               \\
\vdots     &     & \ddots &         &\vdots     \\
   0 &     & \cdots & 0       & a_{n-1}         \\
   0 &     & \cdots & - a_{n-1} & 0 
\end{matrix}
\right)
\,.
\end{equation}
%---------------------------------------
If $O(t)$ is the orthogonal matrix solving the equation
\[
{d\over dt}O=BO\,,\qquad O(0)= \; \mbox{Identity} \,,
\]
then from (\ref{lax}) we have
\[
{d\over dt}(O^{-1}LO)=0\,.
\]

Thus, $O^{-1}LO=L(0)$, i.e., $L(t)$ is related to $L(0)$ by 
conjugation with an orthogonal matrix and thus 
the eigenvalues of $L$, 
which are real and distinct, are preserved along the flow. 
This is enough to show that this system is explicitly solvable 
or integrable.  Equivalently, after fixing the center of mass, 
i.e., setting $b_1 + \cdots + b_n = 0$, the $n-1$ integrals 
in involution whose differentials are linearly independent on 
an open dense set of phase space $\{(a_1, \ldots, a_{n-1}, 
b_1, \ldots, b_n) \mid b_1 + \cdots + b_n = 0\}$ are 
$\operatorname{Tr}L^2, \ldots, \operatorname{Tr}L^n$.

%%%%%%%%%%%%%%%%%%%
%%%%%%%%%%%%%%%%%%%
\subsection{Lie algebra integrability of the Toda  system}

Let us quickly recall the well-known Lie algebraic approach to 
integrability of the Toda lattice. Let $\mathfrak{g}$ be a Lie algebra 
with an invariant non-degenerate bilinear symmetric form 
$\left\langle\,,\right\rangle$, i.e.,
$\left\langle[\xi, \eta], \zeta \right\rangle = 
\left\langle\xi, [\eta, \zeta] \right\rangle$ for all $\xi, \eta,\zeta
\in \mathfrak{g}$ and $\left\langle \xi, \cdot \right\rangle = 0$
implies $\xi=0$. Suppose that $\mathfrak{k}, \mathfrak{s} \subset 
\mathfrak{g}$ are Lie subalgebras and that, \textit{as vector spaces},
$\mathfrak{g} = \mathfrak{k}\oplus\mathfrak{s}$. Let $\pi_{\mathfrak{k}}
: \mathfrak{g}\rightarrow \mathfrak{k}$, $\pi_{\mathfrak{s}}:\mathfrak{g}
\rightarrow \mathfrak{s}$ be the two projections induced by this vector
space direct sum decomposition. Since $\mathfrak{g}\ni
\xi \stackrel{\sim}\longmapsto \left\langle \xi, \cdot \right\rangle
\in \mathfrak{g}^\ast$ is a vector space isomorphism, it naturally
induces the isomorphisms $\mathfrak{k}^\perp \cong \mathfrak{s} ^\ast$, 
$\mathfrak{s}^\perp \cong \mathfrak{k}^\ast$. By non-degeneracy of
$\left\langle\,, \right\rangle$, we have $\mathfrak{g} = 
\mathfrak{s}^\perp \oplus \mathfrak{k}^\perp$; denote by 
$\pi_{\mathfrak{k}^\perp}: \mathfrak{g}\rightarrow \mathfrak{k}^\perp$, 
$\pi_{\mathfrak{s}^\perp}:\mathfrak{g}\rightarrow \mathfrak{s}^\perp$
the two projections induced by this vector space direct sum decomposition.
In particular, $\mathfrak{g}$, $\mathfrak{s}^\perp$, $\mathfrak{k}^\perp$
all carry natural Lie-Poisson structures. The (-)Lie-Poisson bracket of 
$\mathfrak{s}^\ast \cong \mathfrak{k}^\perp$ is given by
\begin{equation}
\label{pb_general}
\{\varphi, \psi\}(\xi) = - \left\langle \xi, 
\left[\pi_{\mathfrak{s}}\nabla\varphi(\xi) , 
\pi_{\mathfrak{s}}\nabla\psi(\xi) \right] \right\rangle, \qquad 
\xi\in \mathfrak{k}^\perp,
\end{equation}
where $\varphi, \psi: \mathfrak{k}^\perp \rightarrow\mathbb{R}$ are
any smooth functions, extended arbitrarily to smooth functions, also
denoted by $\varphi$ and $\psi$, on
$\mathfrak{g}$ and $\nabla\varphi$, $\nabla\psi$ are the gradients of 
these arbitrary extensions relative to $\left\langle\,, \right\rangle$.
This formula follows from the fact that the gradient on 
$\mathfrak{k}^\perp$ of $\varphi|_{\mathfrak{k}^\perp}$, which is an 
element of $\mathfrak{s}$ due to the isomorphism $\mathfrak{k}^\perp 
\cong \mathfrak{s} ^\ast $,  equals $\pi_{\mathfrak{s}}\nabla\varphi$.
Thus, the Hamiltonian vector field of $\psi \in 
C ^{\infty}(\mathfrak{k}^\perp)$, given by 
$\dot{\varphi} = \{\varphi, \psi\}$ for any $\varphi \in 
C ^{\infty}(\mathfrak{k}^\perp)$, has the expression
\begin{equation}
\label{ham_vf_general}
X_\psi(\xi) = - \pi_{\mathfrak{k}^\perp}
\left[\pi_\mathfrak{s}  \nabla\psi( \xi) , \xi\right], 
\qquad \xi\in \mathfrak{k}^\perp
\end{equation}
with the same conventions as above.

If $\psi \in C ^{\infty}(\mathfrak{g})$ is invariant, i.e., $[\nabla \psi(\zeta), \zeta] =0$ for all $\zeta\in \mathfrak{g}$, then  \eqref{ham_vf_general} simplifies to
\begin{equation}
\label{ham_vf_inv}
X_\psi(\xi) = \left[\pi_\mathfrak{k}  \nabla\psi( \xi) , \xi\right]
= - \left[\pi_\mathfrak{s}  \nabla\psi( \xi) , \xi\right],
\qquad \xi\in \mathfrak{k}^\perp.
\end{equation}
The Adler-Kostant-Symes Theorem (see \cite{Adler1979}, \cite{Kostant1979},
\cite{Symes1980a, Symes1980b}, and \cite{Ratiu1980b} for many theorems
of the same type) states that if $\varphi$ and $\psi$ are both invariant
functions on $\mathfrak{g}$, then $\{\varphi, \psi\} = 0$ on 
$\mathfrak{k}^\perp$ which is equivalent to the commutation of the 
flows of the Hamiltonian vector fields  \eqref{ham_vf_inv}.

Suppose that $G = KS$, where $G$ is a Lie group with Lie algebra 
$\mathfrak{g}$ and $K, S \subset G$ are closed subgroups with Lie
algebras $\mathfrak{k}$ and $\mathfrak{s} $, respectively. The writing
$G=KS$ means that each element $g\in G$ can be uniquely decomposed as
$g = ks$, where $k \in K$ and $s\in S$ and that this decomposition 
defines a smooth diffeomorphism $K \times S \approx G$. The coadjoint
action of $S$ on $\mathfrak{s}^\ast$ has the following expression, if
$\mathfrak{s}^\ast$ is identified with $\mathfrak{k}^\perp$ via
$\left\langle\,,\right\rangle$: if $s \in S$, $\xi\in \mathfrak{k}^\perp$, then $s \cdot \xi= \pi_{\mathfrak{k}^\perp} \operatorname{Ad}_{s}\xi$,
where $\operatorname{Ad}_s\xi$ is the adjoint action in $G$ of the
element $s \in S \subset G$ on $\xi\in \mathfrak{k}^\perp \subset \mathfrak{g}$.

For the Toda lattice \eqref{lax}, this general setup applies in the following way. Let $G = \operatorname{GL}(n, \mathbb{R})$, $K = 
\operatorname{SO}(n)$, 
$S=\{\text{invertible lower triangular matrices}\}$, $G=KS$ is the
Gram-Schmidt orthonormalization process, 
$\mathfrak{g} = \mathfrak{gl}(n, \mathbb{R})$, $\mathfrak{k} = 
\mathfrak{so}(n)$, $\mathfrak{s} = \{\text{lower triangular matrices}\}$,
$\left\langle\xi, \eta\right\rangle: = 
\operatorname{Tr}(\xi\eta)$ for
all $\xi, \eta\in \mathfrak{gl}(n, \mathbb{R})$, 
$\mathfrak{k}^\perp =\mathfrak{sym}(n)$ the vector space of 
symmetric matrices, and $\mathfrak{s}^\perp = \mathfrak{n}$, 
the nilpotent Lie algebra of strictly lower triangular 
matrices. The set of matrices $L$ in \eqref{lb_matrices} 
is a union of $S$-coadjoint orbits parametrized by
the value of the trace; for example, the set of trace zero 
matrices $L$ of the form \eqref{lb_matrices} equals the
$S$-coadjoint orbit through the symmetric matrix that
has everywhere zero entries with the exception of the upper and lower
first diagonals where all entries are equal to one. Thus, the Toda lattice is a 
Poisson system  whose restriction to a symplectic leaf is a classical 
Hamiltonian system with $n-1$ degrees of freedom. The Hamiltonian
of the Toda lattice is $\frac{1}{2}\operatorname{Tr}L^2$ and the 
$f_k(L) : =\frac{1}{k}\operatorname{Tr}L^k$,  $k=1, \ldots, n-1$ are the
$n-1$ integrals in involution (by the Adler-Kostant-Symes Theorem) and are
generically independent.

%%%%%%%%%%%%%%%%%%%
%%%%%%%%%%%%%%%%%%%
\subsection{The Toda system as a double bracket equation} 

If $N$ is the matrix $\operatorname{diag}\{1,2,\dots ,n\}$,
the Toda equations \eqref{lax} may be written in the 
double bracket form \eqref{dbflow} for $B:=[N,L]$. This was 
shown in \cite{Bloch1990}; the consequences of this fact 
were further analyzed for general compact Lie algebras in 
\cite{BlBrRa1990}, \cite{BlBrRa1992}, and \cite{BlFlRa1990}. 
As shown in Proposition \ref{double_bracket_prop}, the
double bracket equation, with $L$ replaced by ${\rm i}L$ and  
$N$ by ${\rm i}N$, restricted to a level set of the 
integrals described above, i.e., restricted to a generic
adjoint orbit of $\operatorname{SU}(n)$, is the gradient 
flow of the function ${\rm Tr} LN$ with respect to the normal 
metric; see \cite{BlFlRa1990} for this approach.

This observation easily implies that the flow tends
asymptotically to a diagonal matrix with the eigenvalues of 
$L(0)$ on the diagonal and ordered according to magnitude, 
recovering the result of \cite{Moser1975}, \cite{Symes1982}, 
and \cite{DeNaTo1983}.
 
%%%%%%%%%%%%%%%%%%%
%%%%%%%%%%%%%%%%%%%
\subsection{Riemannian metrics on $\mathcal{O}$} 

Now,  we recall that, in  addition to the normal metric on an 
adjoint orbit, there are other natural $G_u$-invariant metrics: the induced 
and the group invariant K\"ahler
metrics (as discussed in \cite[\S4]{Atiyah1982}, \cite{AtPr1983}, and \cite[Chapter 8]{Besse2008}).

Firstly, there is the \textit{induced metric} $b$ on 
$\mathcal{O}$, defined by $b: = \iota^\ast\left(-\kappa(~~,~~)
\right)$, where $\iota: \mathcal{O}\hookrightarrow 
\mathfrak{g}_u$ is the inclusion and $\left\langle~~,~~ \right\rangle:=
-\kappa(~~,~~)$ is thought of as a  constant 
Riemannian metric on $\mathfrak{g}_u$. Therefore,
\begin{equation}
\label{induced_metric}
b(L)([L,X], [L,Y]): = \left\langle [L,X], [L,Y] \right\rangle
\end{equation}
for any $L \in \mathcal{O}$, $X,Y \in \mathfrak{g}_u$. The induced
metric on $\mathcal{O}$ is also $G_u$-invariant.

Secondly, there are the $G_u$-invariant K\"ahler metrics on
$\mathcal{O}$ compatible with the natural complex structure (of
course, induced by the complex structure of $G$).
These are in bijective correspondence (by the transgression
homomorphism) with the set of $G_u$-invariant sections of the
trivial vector bundle over $\mathcal{O}$ whose fiber at $L \in 
\mathcal{O}$ is the center of $\ker \left(\operatorname{ad}_L\right)$
and whose scalar product with all positive roots is positive
(\cite[Proposition 8.83]{Besse2008}). Among these, there is
the $G_u$-invariant K\"ahler metric $b_2$ which is compatible 
with both the natural complex structure on $\mathcal{O}$ and has as 
imaginary part the orbit symplectic structure; $b_2$ is called the 
\textit{standard K\"ahler metric} on $\mathcal{O}$.

The $G_u$-invariant Riemannian metrics on a maximal dimensional orbit 
$\mathcal{O}$ are completely determined by $T$-invariant inner 
products on the direct sum of the two dimensional root spaces of
$\mathfrak{g}_u$, which is the tangent space to $\mathcal{O}$
at the point $L_0 \in \mathfrak{t}$ in the interior of the 
positive Weyl chamber;
recall that $\mathcal{O}$ intersects the positive Weyl chamber in
a unique point. The negative of the Killing form induces on each such
2-dimensional space an inner product. This inner product, left translated
at all points of $\mathcal{O}$ by elements of $G_u$, yields the normal 
metric on $\mathcal{O}$. Any other $G_u$-invariant inner product on
$\mathcal{O}$ is obtained by left translating at all points of 
$\mathcal{O}$ the inner product on this direct sum of 2-dimensional
root spaces obtained by multiplying in each 2-dimensional summand
the inner product with a positive real constant.

Since $L_0$ lies in the interior of the positive Weyl chamber 
(because $\mathcal{O}$ is maximal dimensional),
$\alpha(L_0) >0$ for all positive roots $\alpha$ of $\mathfrak{g}_u$.
Then the constant by which the natural inner product on the 2-dimensional
root space needs to be multiplied in order to get the standard K\"ahler metric is $\alpha(L_0)$, whereas to get the induced metric, it is
$\alpha(L_0)^2$ (\cite[Remark 2 in \S4]{Atiyah1982}). We can formulate
this differently, as in \cite{BlFlRa1990}. Since, by \eqref{induced_metric} and \eqref{normal_metric},
\begin{align*}
b(L)([L,X], [L,Y]) &= \left\langle [L,X], [L,Y] \right\rangle
= \left\langle [L,X^L], [L,Y^L] \right\rangle
= \left\langle -[L,[L,X^L]], Y^L \right\rangle
= \left\langle -[L,[L,X^L]]^L, Y^L \right\rangle \\
&= \left\langle - \operatorname{ad}_L^2[L,X],[L,Y] \right\rangle_{\rm normal}
\end{align*}
we have 
\begin{equation}
\label{induced_normal}
b(L)([L,X], [L,Y]) = b_1(L)(\mathcal{A}(L)^2[L,X], [L,Y]),
\end{equation}
where we denote now by $b_1$ the normal metric and $\mathcal{A}(L) := 
\sqrt{\left({\rm i}\operatorname{ad}_L\right)^2}$ is the positive 
square root of $\left({\rm i}\operatorname{ad}_L\right)^2 
= - \operatorname{ad}_L ^2 = \mathcal{A}(L)^2$. The standard K\"ahler metric on 
$\mathcal{O}$ is then given by 
\begin{equation}
\label{kahler_normal}
b_2(L)[L,X], [L,Y]) = b_1(\mathcal{A}(L)[L,X], [L,Y]).
\end{equation}
Note that, as opposed to the normal and induced metrics which have
explicit expressions, the standard K\"ahler metric on $\mathcal{O}$
requires the spectral decomposition of $\mathcal{A}(L)$ at any
point $L \in \mathcal{O}$. Or, as explained above, one expresses
it at the point $L_0$ in the positive Weyl chamber in terms of the
positive roots and then left translates the resulting inner product
at any point of $\mathcal{O}$. The normal metric does not depend on
the operators $\mathcal{A}(L)$, whereas the standard K\"ahler and
induced metrics do.

%%%%%%%%%%%%%%%%%%%
%%%%%%%%%%%%%%%%%%%
%%%%%%%%%%%%%%%%%%%
\section{Gradient flows on the loop group of the circle}
\label{gradients_section}

In this section we introduce three weak Riemannian metrics on the
subgroup of average zero functions of the connected component of the
loop group $\widetilde{\operatorname{L}}(S^1)$ of the circle, analogous 
to the normal, standard K\"ahler, and induced metrics on adjoint orbits of compact semisimple Lie groups. Of course, we
shall not work on adjoint orbits of this group because they degenerate
to points, $\widetilde{\operatorname{L}}(S^1)$ being a commutative group. 
Then we shall compute the gradient flows for these three metrics.

%%%%%%%%%%%%%%%%%%%
%%%%%%%%%%%%%%%%%%%
\subsection{The loop group of $S^1$} 

Recall (e.g., \cite{PrSe1986}) that the loop group 
$\widetilde{\operatorname{L}}(S^1)$ of the circle $S^1$ consists of 
smooth maps of $S^1$ to $S^1$. With pointwise multiplication, 
$\widetilde{\operatorname{L}}(S^1)$ is a commutative group. Often, 
elements of $\widetilde{\operatorname{L}}(S^1)$ are written as 
$e^{{\rm i}f}$, where $f \in \widetilde{\operatorname{L}}(\mathbb{R}):=
\left\{g : [-\pi,\pi] \to \mathbb{R}\mid g\text{ is } C^{\infty},\; 
g(\pi) = g(-\pi) + 2n\pi, \text{ for some } n \in \mathbb{Z}\right\}$; 
$n$ is the \textit{winding number} of the closed curve $[-\pi, \pi] \ni
t \mapsto e^{\mathrm{i}g(t)} \in S^1$ about the origin.
More precisely, there is an exact sequence of groups
\begin{eqnarray*}
\begin{array}{ccccccccccc}
0& \longrightarrow & \mathbb{Z}& \longrightarrow &\widetilde{\operatorname{L}}(\mathbb{R})& \stackrel{\widetilde{\operatorname{exp}}}\longrightarrow 
&\widetilde{\operatorname{L}}(S^1)& \longrightarrow &\mathbb{Z} &\longrightarrow &0\\
&&n&\longmapsto& 2 \pi n;\quad f &\longmapsto& e^{\mathrm{i}f}&\longmapsto &\frac{f(\pi) - f(-\pi)}{2 \pi}&&
\end{array}
\end{eqnarray*}
which shows that $\operatorname{ker}{\widetilde{\exp}} = \mathbb{Z}$ and $
\operatorname{coker}{\widetilde{\exp}} = \{0\}$. Thus the 
connected components of 
$\widetilde{\operatorname{L}}(S^1)$ are indexed by the winding number. 
The connected component of the identity 
$\widetilde{\operatorname{L}}(S^1)_0$ consists of loops
with winding number zero about the origin. 

If one insists on working with smooth loops, then one can consider 
$\widetilde{\operatorname{L}}(S^1)$ and 
$\widetilde{\operatorname{L}}(S^1)_0$
as Fr\'echet Lie groups either in the convenient calculus of
\cite{KrMi1997} or in the tame category of \cite{Hamilton1982}. 

Alternatively, one can work with loops $e^{{\rm i}f}$
for $f: [-\pi, \pi ] \rightarrow \mathbb{R}$ of Sobolev
class $H^s$, where $s \geq 1$ (or appropriate $W ^{s,p}$ or H\"older 
spaces). By standard theory (see, e.g., \cite{Palais1968} or 
\cite{AdFo2003}), it is checked that $\widetilde{\mathrm{L}}(S^{1})$ 
is a Hilbert Lie group (see, e.g., \cite{Bourbaki1998} or 
\cite{Neeb2004}). We shall not add the index $s$ on 
$\widetilde{\operatorname{L}}(\mathbb{R})$ and 
$\widetilde{\operatorname{L}}(S^1)$; from now on we
work exclusively in this category of $H^s$ Sobolev class maps and loops. 
A simple proof of the fact 
that $\widetilde{\operatorname{L}}(\mathbb{R})$  is a Hilbert Lie
group was given to us by K.-H. Neeb. First, note that 
$\widetilde{\operatorname{L}}(\mathbb{R})$
is a closed additive subgroup of the Hilbert space $H^s(\mathbb{R})
:=\{h:\mathbb{R}\rightarrow \mathbb{R}\mid h\text{ of class } H^s\}$.
Second, $\widetilde{\operatorname{L}}(\mathbb{R}) 
= \widetilde{\operatorname{L}}(\mathbb{R})_0 \times \mathbb{Z}$ 
as topological groups, where 
$\widetilde{\operatorname{L}}(\mathbb{R})_0 : = \{g \in 
\widetilde{\operatorname{L}}(\mathbb{R})\mid g(\pi) = g(-\pi)\}$
is the closed vector subspace of $H^s(\mathbb{R})$ consisting
of periodic functions; hence it is an additive Hilbert Lie group.
Therefore, there is a unique Hilbert Lie group structure on
$\widetilde{\operatorname{L}}(\mathbb{R})$ for which 
$\widetilde{\operatorname{L}}(\mathbb{R})_0$ is the connected
component of the identity. For general criteria that characterize
Lie subgroups in infinite dimensions, see \cite[Theorem IV.3.3]{Neeb2006}
(even for certain classes of Lie groups modeled on locally convex
spaces). Third, since $\widetilde{\operatorname{exp}}: \widetilde{\operatorname{L}}(\mathbb{R}) \rightarrow 
\widetilde{\operatorname{L}}(S^1)$ maps bijectively 
each connected component of 
$\widetilde{\operatorname{L}}(\mathbb{R})$ to a connected 
component of $\widetilde{\operatorname{L}}(S^1)$, it induces a Hilbert
Lie group structure on $\widetilde{\operatorname{L}}(S^1)$.

The commutative Hilbert Lie algebra of $\widetilde{\operatorname{L}}(S^1)$ 
is clearly $H^s(S^1, \mathbb{R}): = \{u:S^1\rightarrow \mathbb{R}\mid u
\text{ of class } H^s\}$, the space of periodic $H^s$ maps, 
and the exponential map $\operatorname{exp}:
H^s(S^1, \mathbb{R}) \rightarrow \widetilde{\operatorname{L}}(S^1)$
is given by $\operatorname{exp}(u)(\theta) = e^{ {\rm i}u(\theta)}$,
where $\theta \in \mathbb{R}/2\pi\mathbb{Z} = S^1$.

%%%%%%%%%%%%%%%%%%%
%%%%%%%%%%%%%%%%%%%
\subsection{The based loop group of $S^1$} 

The inner product
on the Hilbert space $L^2(S^1)$ of $L^2$ real valued functions on 
$S^1$ is defined by
\[
\left\langle f, g \right\rangle: = \pint \theta\, 
f(\theta) g(\theta)\,, \quad f,g \in L^2(S^1).
\]

Following \cite{Pressley1982} and \cite{AtPr1983}, we introduce the
closed Hilbert Lie subgroup $\operatorname{L}(S^1): = \{\varphi \in 
\widetilde{\operatorname{L}}(S^1)\mid \varphi(1) = 1\}$ of 
$\widetilde{\operatorname{L}}(S^1)$ whose closed commutative
Hilbert Lie algebra is $\operatorname{L}(\mathbb{R}) : = 
\{u \in H^s(S^1, \mathbb{R})\mid u(1) = 0\}$. The exponential map 
$\exp: \operatorname{L}(\mathbb{R}) \ni u \mapsto e^{{\rm i}u} \in 
\operatorname{L}(S^1)$ is a Lie group isomorphism 
(with $\operatorname{L}(\mathbb{R})$ thought of as a commutative group
relative to addition), a fact that will play a
very important role later on (see also \cite[page 151, \S8.9]{PrSe1986}).

There is a natural 
2-cocycle $\omega$ on $\operatorname{L}(\mathbb{R})$, namely
\begin{equation}
\label{cocycle}
\omega(u,v) : = 
\pint \theta\,  u'(\theta) v(\theta) = 
\left\langle u', v \right\rangle\,,
\end{equation}
where $u':=d u/d\theta$. Therefore, there is a central extension of Lie algebras
\[
0 \longrightarrow \mathbb{R}\longrightarrow
\widehat{\operatorname{L}(\mathbb{R})}\longrightarrow
\operatorname{L}(\mathbb{R})\longrightarrow 0
\]
which, as shown in \cite{Segal1981}, integrates to a central extension of 
Lie groups
\[
1 \longrightarrow S^1\longrightarrow
\widehat{\operatorname{L}(S^1)}\longrightarrow
\operatorname{L}(S^1)\longrightarrow 1.
\]
The ``geometric duals'' of $\operatorname{L}(\mathbb{R})$ and 
$\widehat{\operatorname{L}(\mathbb{R})} = \mathbb{R} \oplus 
\operatorname{L}(\mathbb{R})$ are themselves, relative to the weak 
$L^2$-pairing. It turns out that 
the coadjoint action of $\widehat{\operatorname{L}(S^1)}$ on 
$\widehat{\operatorname{L}(\mathbb{R})}$ preserves $\{1\} \oplus 
\operatorname{L}(\mathbb{R})$ so that, as usual, the coadjoint action of 
$\widehat{\operatorname{L}(S^1)}$ on $\operatorname{L}(\mathbb{R})$ is an
affine action which, in this case, because the group is commutative,
equals 
\[
\operatorname{Ad}_{e^{{\rm i}f}}^\ast \mu = \frac{f'}{f} = \left(\log|f| \right)'
\qquad e^{{\rm i}f} \in \operatorname{L}(S^1),\quad \mu \in 
\operatorname{L}(\mathbb{R}).
\]
Thus, the orbit of the constant function 0 is 
$\widehat{\operatorname{L}(S^1)}/S^1$
(where the denominator is thought of as constant loops), i.e., it equals
$\operatorname{L}(S^1)$. Therefore, every element  
$u \in \operatorname{L}(\mathbb{R})$ of its Lie algebra has, in  
Fourier representation, vanishing zero order Fourier coefficient , i.e.,
$\widehat{u}(0) = 0$.

Thus, the based loop group is a coadjoint orbit
of its natural central extension and, according  to 
\S\ref{metrics_on_orbits}, has three distinguished weak Riemannian 
metrics. These were computed explicitly in \cite{Pressley1982}, 
\cite{AtPr1983}, \cite{PrSe1986}; we recall them below.

%%%%%%%%%%%%%%%%%%%
%%%%%%%%%%%%%%%%%%%
\subsection{$\operatorname{L}(S^1)$ as a weak K\"ahler manifold}

Note that on $\operatorname{L}(\mathbb{R})$, the cocycle \eqref{cocycle} 
is weakly non-degenerate. Therefore, left (or right) translating it at
every point of the group $\operatorname{L}(S^1)$ yields a weakly 
non-degenerate closed two-form, i.e., a symplectic form. Thus, as expected,
since it is a coadjoint orbit, the Hilbert Lie group 
$\operatorname{L}(S^1)$ carries an invariant symplectic form whose 
value at the identity element 1 (the constant loop equal to 1) is
given by \eqref{cocycle}.
\smallskip

Now we introduce the \textit{Hilbert transform} on the circle
\begin{equation}
\label{Hilbert_transform}
\mathcal{H}u(\theta):=\dpint s\,  u(s)
\cot\left(\frac{\theta-s}{2} \right)
= \dpint s\,  u(\theta-s)\cot\left(\frac{s}{2}\right)
: = \lim_{\varepsilon\to 0+}\frac{1}{\pi}\int_{\varepsilon \leq |s|\leq \pi}\!{\rm d}s \, u(\theta-s)\cot\left(\frac{s}{2}\right)
\end{equation}
for any $u \in L^2(S^1)$, where $\dashint$ denotes the Cauchy principal 
value. We adopt here the sign conventions in 
\cite[Formulas (3.202) and (6.38), Vol. 1]{King2009}. If 
$u \in L^2(S^1)$, then $\mathcal{H}u \in L^2(S^1)$ and it is defined for 
almost every $\theta\in [-\pi, \pi]$ (Lusin's Theorem, 
\cite[\S6.19, Vol. 1]{King2009}). The Hilbert transform has the 
following remarkable properties that will be used later on: 

$\bullet$ If $u( \theta) = \sum_{n=-\infty}^\infty \widehat{u}(n) 
e^{{\rm i}n \theta} \in L^2(S^1)$, where $\widehat{u}(n): = 
\pint \theta\,  u(\theta) e^{- {\rm i}n \theta}$, so 
$\overline{\widehat{u}(n)} = \widehat{u}(-n)$ since $u$ is real valued, then
\begin{equation}
\label{Hilbert_transform_Fourier}
\mathcal{H}u(\theta) = - {\rm i}\sum_{n=-\infty}^\infty \widehat{u}(n) 
\operatorname{sign}(n)\, e^{{\rm i}n \theta} \in L^2(S^1)
\end{equation}
which follows from the identity $\widehat{\mathcal{H} f}(n) = 
- {\rm i} \widehat{f}(n) \operatorname{sign}(n)\, $ 
(\cite[Formulas (6.100) or (6.124), Vol. 1]{King2009}). Here, 
$\operatorname{sign}(n) = 1$ if $n \in \mathbb{N}$, $\operatorname{sign}(n) = -1$ if $n \in -\mathbb{N}$, and $\operatorname{sign}(0) = 0$. Note that
$\mathcal{H}u$ is also real valued since 
$\widehat{u}(n) \operatorname{sign}(n) = 
- \widehat{u}(-n) \operatorname{sign}(-n)$.  The formula above
implies that (\cite[Formula (6.126), Vol. 1]{King2009})
\[
\int_{- \pi}^\pi\!  {\rm d}s\,  \mathcal{H}u(s) = 0.
\]

$\bullet$ For every $u \in L^2(S^1)$, we have the orthogonality property 
(\cite[Formula (6.127), Vol. 1]{King2009}):
\[
\left\langle u,\mathcal{H}u \right\rangle = 0.
\]

$\bullet$ Take the orthonormal Hilbert basis $\left\{\varphi_n(\theta): =
e^{ {\rm i}n \theta}\mid n \in \mathbb{Z} \right\}$ of $L^2(S^1)$. Then (\cite[Formula (6.131), Vol. 1]{King2009}):
\[
\mathcal{H}\varphi_n(\theta) = -{\rm i} \operatorname{sign}(n)\,  \varphi_n( \theta), \quad \text{for all}\quad n \in \mathbb{Z}.
\]
So, the eigenvalues of $\mathcal{H}$ are: $-{\rm i}$ for all $n>0$, 
${\rm i}$ for all $n<0$, and $0$ if $n=0$.

$\bullet$ If $u,v \in L^2(S^1)$ then 
(\cite[Formula (6.99), Vol. 1]{King2009})
\[
\left\langle u, v \right\rangle = 
\frac{1}{4 \pi^2}\left(\int_{-\pi}^\pi \!{\rm d}s\, u(s)\right)
\left(\int_{-\pi}^\pi\!{\rm d}s\,  v(s)\right) + 
\left\langle \mathcal{H}u, \mathcal{H}v \right\rangle
\]
and hence (\cite[Formula (6.97), Vol. 1]{King2009})
\[
\|u\|_{L^2(S^1)}^2 = \left(\frac{1}{2 \pi}\int_{-\pi}^\pi\!{\rm d}s\,  u(s)  \right)^2
+ \| \mathcal{H}u\|_{L^2(S^1)}^2
\]
for any $u \in L^2(S^1)$. This shows that $\| \mathcal{H}u\|_{L^2(S^1)}^2
\leq \|u\|_{L^2(S^1)}^2$ and the constant 1 is the best possible
(\cite[Formulas (6.167) and (6.168), Vol. 1]{King2009}). In
particular, if the average of $u$ is zero, then $\mathcal{H}$ is an
isometry of $L^2(S^1)$.

$\bullet$ The Hilbert transform is skew-adjoint relative to 
the $L^2(S^1)$-inner product, i.e., $\mathcal{H}^\ast  = 
- \mathcal{H}$ (\cite[Formula (6.98) or (6.106), Vol. 1]{King2009}).

$\bullet$ For any $u \in L^2(S^1)$ we have 
(\cite[Formula (6.34),  (6.82), or (6.156), Vol. 1]{King2009}):
\[
\mathcal{H}^2 u(\theta) = - u(\theta) + 
\frac{1}{2 \pi}\int_{-\pi}^\pi\!{\rm d}s\,  u(s)  = 
-u(\theta) + \widehat{u}(0).
\]

$\bullet$ For any $u \in H^s(S^1)$ with $s\geq 0$ we have 
$\mathcal{H}u 
\in H^s(S^1)$; this is an immediate consequence of 
\eqref{Hilbert_transform_Fourier}. If $s\geq 1$,
then $\mathcal{H}u' = (\mathcal{H}u)'$, i.e., $\mathcal{H} 
\circ\frac{d}{d \theta} = \frac{d}{d \theta}\circ\mathcal{H}$
on  $H ^s(S ^1)$ with $s \geq 1$.
\medskip

Using these properties, if $u( \theta) = 
\sum_{n=-\infty}^\infty 
\widehat{u}(n) e^{{\rm i}n \theta} \in H^1(S^1)$, then 
$u'(\theta) = \sum_{n=-\infty}^\infty \widehat{u}(n) {\rm i}n
e^{{\rm i}n \theta} \in L^2(S^1)$ and hence
\begin{equation}
\label{Hilbert_derivative}\left(\mathcal{H}u'\right)(\theta) = \left(\mathcal{H}u \right)'(\theta) = 
\left(- {\rm i}\sum_{n=-\infty}^\infty \widehat{u}(n) 
\operatorname{sign}(n)\, e^{{\rm i}n \theta} \right)'
=\sum_{n=-\infty}^\infty|n|\widehat{u}(n)e^{in\theta}\,.
\end{equation}
On the other hand, if $v \in H^2(S^1)$, then
\begin{equation}
\label{second_derivative}
-\frac{d^2}{d\theta^2}v(\theta)=
\sum_{n=-\infty}^\infty n^2\widehat{v}(n)e^{in\theta}
\end{equation}
and hence if $u \in H^1(S^1)$,
\begin{equation}
\label{square_root}
\left(-\frac{d^2}{d\theta^2}\right)^{\frac{1}{2}}u(\theta)=
\sum_{n=-\infty}^\infty |n|\widehat{u}(n)e^{in\theta}=
(\mathcal{H}u')(\theta) = \left(\left(\mathcal{H}\circ 
\frac{d}{d\theta}\right) u \right)(\theta)
\end{equation}
by \eqref{Hilbert_derivative}. By the previous properties we 
have $\left(\mathcal{H} \circ {d}/{d\theta}\right)^2 = 
- {d^2}/{d\theta^2}$, as expected; note that the the extra 
term, which is the zero order Fourier coefficient, does not 
appear in this case, because the derivative eliminates it.
\medskip

Now,  if $\varphi = e^{{\rm i}f} \in \operatorname{L}(S^1)$, 
i.e., $\varphi(1) = 1$ and $f:[-\pi, \pi] \rightarrow 
\mathbb{R}$ is a periodic function, then $\widehat{f}(0) = 
f(0) = 0$. Similarly, if $u \in 
\operatorname{L}(\mathbb{R})$, i.e., $u(1) = 0$ and we think 
of $u$ as a periodic function $u:[-\pi, \pi] \rightarrow 
\mathbb{R}$, then $\widehat{u}(0) = u(0) = 0$. This, and the 
properties of the Hilbert transform on the circle, imply:  
$\mathcal{H}\left(\operatorname{L}(\mathbb{R}) \right)
\subseteq \operatorname{L}(\mathbb{R})$, $\mathcal{H}$ is 
unitary on $\operatorname{L}(\mathbb{R})$ (relative to the 
$H^s$-inner product),   
$\mathcal{H} \circ \mathcal{H} = - I$ on 
$\operatorname{L}(\mathbb{R})$.
Concretely, the Hilbert transform on 
$\operatorname{L}(\mathbb{R})$ has the form: 
\[
u(\theta) = \sum_{n\in \mathbb{Z}\setminus\{0\}}\widehat{u}(n) 
e^{{\rm i}n \theta} \in \operatorname{L}(\mathbb{R}) \quad 
\Longrightarrow
\quad 
\mathcal{H}u(\theta) = - {\rm i}\sum_{n\in \mathbb{Z}\setminus
\{0\}}\widehat{u}(n) \operatorname{sign}(n)\, 
e^{{\rm i}n \theta}\in \operatorname{L}(\mathbb{R}).
\]
Thus, $\mathcal{H}$ defines  the structure of a complex 
Hilbert space on $\operatorname{L}(\mathbb{R})$, relative to 
the $H^s$ inner product, $s \geq 1$. Hence, translating 
$\mathcal{H}$ to any tangent space of $\operatorname{L}(S^1)$, 
we obtain an invariant almost complex structure on the Hilbert 
Lie group $\operatorname{L}(S^1)$ which is, in fact, a 
complex structure. For general criteria how to obtain complex 
structures on real Banach manifolds, see \cite{Beltita2005}; 
the argument above is a very special case of these general 
methods.

Finally, $\operatorname{L}(S^1)$ is a K\"ahler manifold, as 
proved in \cite{AtPr1983}. This is immediately seen by noting 
that 
\begin{equation}
\label{kahler}
g(1)(u,v): = \omega(\mathcal{H}u, v) = \sum_{n=-\infty}^
\infty|n|\widehat{u}(n)\widehat{v}(n)
\end{equation}
is symmetric and positive definite and so, by translations, 
defines a weak Riemannian metric on $\operatorname{L}(S^1)$. 
Note that this metric is \textit{not} the $H^s$ metric for 
any $s \geq 1$. In fact, the metric $g$ is incomplete, 
whereas the $H^s$ metric is complete.

Concluding, $(\operatorname{L}(S^1), \omega, g , \mathcal{H})$ 
is a weak K\"ahler manifold and all structures are group 
invariant (see \cite{Pressley1982}, \cite{AtPr1983},
\cite{PrSe1986}).

%%%%%%%%%%%%%%%%%%%
%%%%%%%%%%%%%%%%%%%
\subsection{Weak Riemannian metrics on 
$\operatorname{L}(S^1)$}

The three metrics discussed in \S\ref{metrics_on_orbits} for 
$\operatorname{L}(S^1)$, viewed as a coadjoint orbit of its 
central extension, have been computed by \cite{Pressley1982}. 
We recall here relevant formulas.

The \textit{induced metric} is defined by the natural inner 
product on $\operatorname{L}(\mathbb{R})$, which is the usual 
$L^2$-inner product. Hence, the induced metric is obtained by 
left (equivalently, right) translation of the inner product
\begin{equation}
\label{induced_metric_loop}
b(1)(u,v): = \left\langle u, v \right\rangle = 
\pint t\,  u(t)v(t) 
\end{equation}
for any two functions $u,v \in \operatorname{L}(\mathbb{R})$.
\smallskip

Define the following inner products on 
$\operatorname{L}(\mathbb{R})$:
\begin{align}
\label{kahler_metric_loop}
b_2(1)(u,v) &:= b(1)(u, \mathcal{H}v') =
\left\langle u, \mathcal{H}v' \right\rangle, \quad \text{if} \quad 
u,v \in H^s(S^1), \; s\geq 1\\
\label{normal_metric_loop}
b_1(1)(u,v) &:= b(1)(u', v') =
\left\langle u', v' \right\rangle, \quad \text{if} \quad 
u,v \in H^s(S^1), \; s\geq 1.
\end{align}
Bilinearity and symmetry of $b_1(1)$ and $b_2(1)$ are obvious. 
If $u \in \operatorname{L}(S^1)$, writing $u(\theta) = 
\sum_{n=-\infty}^\infty \widehat{u}(n) e^{ {\rm i}n \theta}$
with $\widehat{u}(0)=0$, we have $u'( \theta) = {\rm i}
\sum_{n=-\infty}^\infty n\widehat{u}(n)e^{in\theta}$. Since 
$\{e^{in\theta}\mid n \in \mathbb{Z}\}$ is an orthonormal 
Hilbert basis of $L^2(S^1)$, we get
\[
b_1(1)(u,u) =\sum_{n=-\infty}^\infty n^2 |\widehat{u}(n)|^2 
\geq 0.
\]
In addition, $b_1(1)(u,u) = 0$ if and only if $\widehat{u}(n) 
=0$ for all $n \neq 0$, i.e., $u(\theta) = \widehat{u}(0)=0$. 
This shows that $b_1(1)$ is indeed an inner product on 
$\operatorname{L}(\mathbb{R})$ which coincides with the $H^1$ 
inner product. Hence, if $\operatorname{L}(\mathbb{R})$ is 
endowed with the $H^s$ topology for $s\geq 1$, this inner 
product is strong if $s=1$ and weak if $s>1$.  Left 
translating this inner product to any tangent space of 
$\operatorname{L}(S^1)$ (endowed with the $H^s$ topology
for $s\geq 1$), yields a Riemannian metric on 
$\operatorname{L}(S^1)$ that is strong for $s=1$ and weak 
for $s>1$. This Riemannian metric is the 
\textit{normal metric} on $\operatorname{L}(S^1)$.

The inner product $b_2(1)$ is identical to $g(1)$ by 
\eqref{kahler}, \eqref{kahler_metric_loop}, and 
\eqref{cocycle}. Thus, translating this inner product to the 
tangent space at every point of the Hilbert Lie group 
$\operatorname{L}(S^1)$, yields the \textit{standard 
K\"ahler metric} $b_2=g$ on $\operatorname{L}(S^1)$, endowed 
with the $H^s$ topology for $s\geq 1$. Note that if 
$u \in \operatorname{L}(S^1)$, then
\[
b_2(1)(u,u) = \sum_{n=-\infty}^\infty|n| |\widehat{u}(n)|^2 
\]
which shows that the K\"ahler metric $b_2$ coincides with 
the $H^{1/2}$ metric and is, therefore, a weak metric on 
$\operatorname{L}(S^1)$.

There are relations similar to \eqref{induced_normal} and 
\eqref{kahler_normal}, namely
\begin{align*}
b(1)(u,v) = b_1(1)(\mathcal{A}^2u, v), \qquad 
b_2(1)(u,v) = b_1(1)(\mathcal{A}u, v),
\end{align*}
where 
\[
(\mathcal{A}^2u)( \theta) = \sum_{n=-\infty}^\infty n^2 
\widehat{u}(n) e^{{\rm i}n \theta}, \qquad 
(\mathcal{A}u)(\theta) = \sum_{n=-\infty}^\infty |n| 
\widehat{u}(n) e^{{\rm i}n\theta}
\]
if $u(\theta) = \sum_{n=-\infty}^\infty \widehat{u}(n)
e^{{\rm i}n\theta}$. However, note that the relation 
involving $\mathcal{A}^2$ requires that $u \in H^s(S^1)$ 
with $s \geq 2$.

%%%%%%%%%%%%%%%%%%%
%%%%%%%%%%%%%%%%%%%
\subsection{Vector fields on $\operatorname{L}(S^1)$ and  $\operatorname{L}(\mathbb{R})$}

Recall that the exponential map
$\exp: \operatorname{L}(\mathbb{R}) \ni u \mapsto e^{{\rm i}u} 
\in \operatorname{L}(S^1)$ is a Lie group isomorphism 
(\cite[page 151, \S8.9]{PrSe1986}). Here, we identified the
Lie algebra of $S^1$ with $\mathbb{R}$, even though, 
naturally, it is the imaginary axis, the tangent space at 
$1 \in S^1$ to $S^1$. This means that care must  be taken 
when carrying out standard Lie  group operations with the 
exponential map, interpreted as the exponential of a purely 
imaginary number. Since such computations
affect our next results, we clarify these statements below.

The tangent space at the identity 1 to $S^1$ is the imaginary 
axis. This is the natural Lie algebra of the Lie group $S^1$ 
and the exponential map is given by 
$\exp: {\rm i} \mathbb{R}\ni ({\rm i}x) \mapsto e^{{\rm i}x} 
\in S^1$. Of course, traditionally, one identifies
${\rm i}\mathbb{R}$ with $\mathbb{R}$ by dividing by 
${\rm i}$ and thinks of the exponential map as 
$\exp: \mathbb{R}\ni x \mapsto e^{{\rm i}x} \in S^1$.  
Unfortunately, this  induces  some problems.
For example, since (left) translation is given by 
$L_{e^{{\rm i}x}} e^{{\rm i}y} : = e^{{\rm i}x}e^{{\rm i}y}$, 
it follows that
\begin{equation}
\label{der_left_translation}
T_1 L_{e^{{\rm i}x}}({\rm i}y) := 
\left.\frac{d}{d\varepsilon}\right|_{\varepsilon=0}
L_{e^{{\rm i}x}} e^{{\rm i}\varepsilon y}
= \left.\frac{d}{d\varepsilon}\right|_{\varepsilon=0}
e^{{\rm i}x}e^{{\rm i}\varepsilon y}
= {\rm i} y e^{{\rm i}x},
\end{equation}
so  the identification of the Lie algebra with $\mathbb{R}$ 
poses no problems and we have, dividing both sides by 
${\rm i}$,
\begin{equation}
\label{der_left_translation_real}
T_1 L_{e^{{\rm i}x}}(y) = y e^{{\rm i}x}.
\end{equation}

However, the definition of
the exponential map for any Lie group $G$ with Lie algebra 
$\mathfrak{g}$, yields
\begin{equation}
\label{exponential_def}
\frac{d}{dt}\exp(t \xi) = T_e L_{\exp(t \xi)} \xi, \quad 
\text{for all} \quad \xi\in \mathfrak{g}. 
\end{equation}
This formula works perfectly well if the Lie algebra of $S^1$ 
is ${\rm i}\mathbb{R}$. Indeed
\[
\frac{d}{dt}e^{t{\rm i}x} = {\rm i}x e^{t{\rm i}x}
\]
which coincides with \eqref{exponential_def} in view of 
\eqref{der_left_translation}. On the other hand, if the Lie 
algebra is thought of as $\mathbb{R}$, i.e., the right hand 
side needs to be divided by ${\rm i}$, then with the 
definition of $\exp(tx) = e^{{\rm i}tx}$ the identity above 
is no longer valid. What we should get is
\[
\frac{d}{dt}\exp(tx) = x \exp(tx) = 
T_1 L_{\exp(tx)} x = xe^{{\rm i}tx}
\]
by \eqref{der_left_translation_real} if $\exp(tx) = 
e^{ {\rm i}tx}$, but the right hand side gives 
${\rm i} x e^{{\rm i}tx}$, as we saw above. In other words,
if the Lie algebra of $S^1$ is thought of as $\mathbb{R}$, 
as  is traditionally done, then we need a formula for the 
derivative of the Lie group exponential map in terms of the 
exponential map of purely imaginary numbers. In view of the 
previous discussion, this formula is
\begin{equation}
\label{der_exp}
\frac{d}{dt}\exp(tx) := 
\frac{1}{{\rm i}}\frac{d}{dt}e^{{\rm i}tx}= 
x e^{{\rm i}tx}.
\end{equation} 

With these remarks in mind, we shall now compute the 
push-forward of a vector field on 
$\operatorname{L}(\mathbb{R})$ to $\operatorname{L}(S^1)$.

\begin{proposition}
\label{passage_prop}
Let $X \in \mathfrak{X}(\operatorname{L}(\mathbb{R}))$ be 
an arbitrary vector field . Then its push-forward to 
$\operatorname{L}(S^1))$ has the expression
\[
\left(\exp_*X\right)\left(e^{{\rm i}u} \right) = 
X(u) e^{{\rm i}u}
\]
for any $u \in \operatorname{L}(\mathbb{R})$.
\end{proposition}

\begin{proof} By the definition of push forward of vector 
fields by a diffeomorphism, we have
\begin{align*}
\left(\exp_*X\right)\left(e^{{\rm i}u} \right) &=
\left(T \exp \circ X \circ \exp ^{-1} \right)
\left(e^{{\rm i}u} \right)
= T_u \exp\left(X(u) \right)
= \left.\frac{d}{d\varepsilon}\right|_{\varepsilon=0}
\exp \left(u+ \varepsilon X(u) \right)\\
&= \left.\frac{d}{d\varepsilon}\right|_{\varepsilon=0}
\exp(u)\exp( \varepsilon X(u))
= \left(\left.\frac{d}{d\varepsilon}\right|_{\varepsilon=0}
\exp( \varepsilon X(u)) \right)\exp(u)
\stackrel{\eqref{der_exp}}= \left(\frac{1}{{\rm i}}
\left.\frac{d}{d\varepsilon}\right|_{\varepsilon=0}
e^{i\varepsilon X(u)}\right) e^{{\rm i}u}\\
&  = X(u)e^{{\rm i}u}
\end{align*}
as stated.
\hfill $\blacksquare$
\end{proof}

%%%%%%%%%%%%%%%%%%%
%%%%%%%%%%%%%%%%%%%
\subsection{The gradient vector fields in the three metrics 
of $\operatorname{L}(S^1)$}

We compute now the gradients of
a specific function using the three metrics.

\begin{theorem}
\label{gradient_thm}
The gradients of the smooth function $H: \operatorname{L}(S^1) \rightarrow \mathbb{R}$ given by 
\[
H\left(e^{{\rm i}f} \right) = 
\frac{1}{4\pi}\int_{-\pi}^\pi\!{\rm d}\theta
\,  f'(\theta)^2 \]
are 
\begin{itemize}
\item[{\rm (i)}] \quad$\nabla^1H\left(e^{{\rm i}f}\right) = 
fe^{{\rm i}f}$ for the normal metric $b_1$; 

\item[{\rm (ii)}] \quad $\nabla H\left(e^{{\rm i}f}\right) = 
-f''e^{{\rm i}f}$ with respect to the induced metric $b$ for $f \in H^s(S^1)$ with $s\geq 2$;

\item[{\rm (iii)}] \quad $\nabla^2 H\left(e^{{\rm i}f}\right) = 
(\mathcal{H}f')e^{{\rm i}f}$ with respect to the weak K\"ahler metric $b_2$.
\end{itemize}
\end{theorem}

\begin{proof}
(i) Since $T_1L_{e^{ {\rm i}f}} u = u e^{{\rm i}f}$ for 
any $u \in \operatorname{L}(\mathbb{R})$ and $e^{{\rm i}f} 
\in \operatorname{L}(S^1)$, invariance of $b_1$ yields
\begin{align*}
b_1(1) \left(e^{-{\rm i}f} \nabla^1H\left(e^{{\rm i}f}\right), u \right)
& = b_1\left(e^{{\rm i}f}\right) \left(\nabla^1H\left(e^{{\rm i}f}\right), ue^{{\rm i}f} \right)
= \mathbf{d}H\left(e^{{\rm i}f}\right)\left(ue^{{\rm i}f} \right)\\
& = \left.\frac{d}{dt}\right|_{t=0}H\left(e^{{\rm i}(f+tu)}\right)
= \left.\frac{d}{dt}\right|_{t=0}\frac{1}{4\pi}\int_{-\pi}^\pi\!{\rm d}\theta\, 
\left(f'(\theta)+tu'(\theta) \right)^2\\
& = \pint \theta\,  f'(\theta)u'(\theta)  
= \left\langle f', u' \right\rangle 
\stackrel{\eqref{normal_metric}} = b_1(1)(f,u)
\end{align*}
which shows that $\nabla^1H\left(e^{{\rm i}f}\right) = fe^{{\rm i}f}$.

(ii) Proceeding as above, using the same notations, and assuming that
$f \in H^s(S^1)$ with $s\geq 2$, we have
\begin{align*}
b(1) \left(e^{-{\rm i}f} \nabla H\left(e^{{\rm i}f}\right), u \right)
& = b\left(e^{{\rm i}f}\right) \left(\nabla H\left(e^{{\rm i}f}\right), ue^{{\rm i}f} \right)
= \mathbf{d}H\left(e^{{\rm i}f}\right)\left(ue^{{\rm i}f} \right)\\
& = \pint \theta\,  f'(\theta)u'(\theta) 
= - \pint \theta\,  f''(\theta)u(\theta) \\
& = \left\langle -f'', u \right\rangle
\stackrel{\eqref{induced_metric}}= b(1)\left(-f'', u\right)
\end{align*}
which shows that $\nabla H\left(e^{{\rm i}f}\right) = -f''e^{{\rm i}f}$.

(iii) This computation uses the isometry property of $\mathcal{H}$ 
relative to the $L^2$ inner product. We have,
\begin{align*}
b_2(1) \left(e^{-{\rm i}f} \nabla^2 H\left(e^{{\rm i}f}\right),u\right)
& = b_2\left(e^{{\rm i}f}\right) \left(\nabla^2 H\left(e^{{\rm i}f}\right), ue^{{\rm i}f} \right)
= \mathbf{d}H\left(e^{{\rm i}f}\right)\left(ue^{{\rm i}f} \right)\\
&= \left\langle f', u' \right\rangle
= \left\langle\mathcal{H}f', \mathcal{H}u' \right\rangle
\stackrel{\eqref{kahler_metric_loop}}=  b_2(1)\left(\mathcal{H}f', u\right)
\end{align*}
which shows that $\nabla^2 H\left(e^{{\rm i}f}\right) = 
(\mathcal{H}f')e^{{\rm i}f}$.
\hfill $\blacksquare$
\end{proof}

Since 
\[
\omega\left(e^{{\rm i}f}\right)
\left(\mathcal{H}\nabla^2H\left(e^{{\rm i}f}\right), 
ue^{{\rm i}f}\right) \stackrel{\eqref{kahler}}=
b_2\left(e^{{\rm i}f}\right)\left(\nabla^2H\left(e^{{\rm i}f}\right), 
ue^{{\rm i}f}\right)
= \mathbf{d}H\left(e^{{\rm i}f}\right) \left(ue^{{\rm i}f}\right)
\]
it follows that the Hamiltonian vector field on 
$\left(\operatorname{L}(S^1), \omega\right)$ for the function $H$ is 
$X_{H} = \mathcal{H}\nabla^2H$. Since $\mathcal{H}$ commutes with
the tangent lift to group translations, Theorem \ref{gradient_thm}(iii)
implies that 
\[
X_{H}\left(e^{{\rm i}f}\right) = \left(\mathcal{H}\nabla^2H \right)
\left(e^{{\rm i}f}\right)
= \mathcal{H}\left(\nabla^2H\left(e^{{\rm i}f}\right)\right)
= \mathcal{H}\left(\left(\mathcal{H}f'\right)e^{{\rm i}f}\right)
= -f'e^{{\rm i}f}.
\]
This proves the first part of the following statement.

\begin{corollary}
\label{hamiltonian_cor}
The Hamiltonian vector field of $H$ relative to the translation invariant 
symplectic form $\omega$ on $\operatorname{L}(S^1)$ whose value at the 
identity element is given by \eqref{cocycle} has the expression 
$X_{H}\left(e^{{\rm i}f}\right) =-f'e^{{\rm i}f}$. Its flow is the
rotation
\[
\left(F_t\left(e^{{\rm i}f} \right)\right)(\theta) = 
e^{-{\rm i}(f(t+ \theta) - f(t))}.
\]
\end{corollary}

\begin{proof}
Since $\operatorname{L}(\mathbb{R}) \ni u\longmapsto 
e^{{\rm i}u} \in \operatorname{L}(S^1)$ is the exponential map and
we think of $\mathbb{R}$ as the Lie algebra of $S^1$ (and not the
imaginary axis), we write ${d}e^{{\rm i}tu}/{dt} = u e^{{\rm i}tu}$
without the factor of ${\rm i}$ in front (see \eqref{der_exp}). The verification that
$F_t$ is indeed the flow of $X_H$ is straightforward:
\begin{align*}
\frac{d}{dt}\left(F_t\left(e^{{\rm i}f} \right)\right)(\theta) 
&= \frac{d}{dt}e^{-{\rm i}(f(t+ \theta) - f(t))} 
= -(f'(t+ \theta) - f'(t))e^{-{\rm i}(f(t+ \theta) - f(t))} \\
& = X_{H} \left(F_t\left(e^{{\rm i}f} \right)\right)( \theta)
\end{align*}
as required.
\hfill $\blacksquare$
\end{proof}

We recover thus \cite[Proposition 3.1]{Pressley1982} (up to a sign which
is due to different conventions calibrating $\omega$, $\mathcal{H}$, 
and $b_2$).

Applying Proposition \ref{passage_prop} to Theorem \ref{gradient_thm},
we get the following result:

\begin{corollary}
\label{h_one_cor}
The three gradient vector fields for the smooth function 
$H_1: \operatorname{L}(\mathbb{R}) \rightarrow 
\mathbb{R}$ given by 
\[
H_1(u)=\frac{1}{4 \pi}\int_{-\pi}^{\pi} \! {\rm d}\theta\, (u')^2
\]
are
\begin{itemize}
\item[{\rm (i)}] \quad$\nabla^1H_1(u) = u$ for the weak inner product $b_1(1)$ defining the normal metric;

\item[{\rm (ii)}] \quad $\nabla H_1(u) = -u''$ for the weak
inner product $b(1)$ defining the induced metric, where for $u \in H^s(\mathbb{R})$ with $s\geq 2$;

\item[{\rm (iii)}] \quad $\nabla^2 H_1(u) = \mathcal{H}u'$ for the 
weak inner product $b_2(1)$ defining the K\"ahler metric.
\end{itemize}
\end{corollary}

Since the exponential map is a Lie group isomorphism and the three metrics
coincide with the respective inner products at the identity, their
left invariance guarantees that the three inner products on 
$\operatorname{L}(\mathbb{R})$ correspond to the three invariant metrics
on $\operatorname{L}(S^1)$. 

Applying Proposition \ref{passage_prop} to Corollary 
\ref{hamiltonian_cor}, we conclude:

\begin{corollary}
\label{hamiltonian_cor_Lie_algebra}
The Hamiltonian vector field of $H_1$ relative to the  
symplectic form $\omega$ given by \eqref{cocycle} has the expression 
$X_{H}(u) =-u'$. Its flow is $\left(F_t(u)\right)(\theta)=u(\theta - t)$.
\end{corollary}
The verification of the statement about the flow is immediate:
\[
\frac{d}{dt}\left(F_t(u)\right)(\theta) = 
\frac{d}{dt}u(\theta-t) = -u'(\theta-t) = 
\left(X_{H}\left(F_t(u)\right)\right)(\theta).
\]

If one is willing to put more stringent hypotheses on the functional,
it is possible to obtain a general result.

\begin{theorem}
\label{general_gradients}
Let $H:\operatorname{L}(S^1) \rightarrow \mathbb{R}$ be a smooth
function {\rm (}with $\operatorname{L}(S^1)$ endowed, as usual, with the
$H^s$ topology for $s\geq 1${\rm )} and assume that the functional 
derivative $\delta H/\delta u \in \operatorname{L}(S^1) $ exists. 
Then the gradient vector fields are  
\begin{itemize}
\item[{\rm (i)}] $\quad \nabla H(u) = \frac{\delta H}{\delta u}$ with respect the weak inner product $b(1)$ defining the induced metric; 
\item[{\rm (ii)}] $\quad\left(\nabla^1H(u) \right)(\theta)= 
-\int_0^\theta\!  \!{\rm d}\varphi\, 
 \left(\int_0^\varphi {\rm d}\psi \, 
 \frac{\delta H}{\delta u}(\psi)\right)$ with respect to the (weak) inner product $b_1(1)$ defining the normal metric, 
provided both $\int_0^\theta \! {\rm d}\varphi \, \frac{\delta H}{\delta u}
(\varphi)$ and $\int_0^\theta \! {\rm d}\varphi\,  \left(\int_0^\varphi  \! {\rm d}\psi\, 
\frac{\delta H}{\delta u}(\psi) \right)$ are periodic;
\item[{\rm (iii)}] $\quad \left(\nabla^2H(u) \right)(\theta)=
-\mathcal{H}\int_0^\theta \!{\rm d}\varphi\, 
 \frac{\delta H}{\delta u} (\varphi)$ with respect to the weak inner product $b_2(1)$ defining the K\"ahler metric, provided
$\int_0^\theta \! {\rm d}\varphi\,  \frac{\delta H}{\delta u}(\varphi)$ is
periodic.
\end{itemize}
\end{theorem}

\begin{proof}
(i) For the inner product $b(1)$ on $\operatorname{L}(S^1)$ defining the
induced metric, if $u,v \in \operatorname{L}(\mathbb{R})$, we have
by periodicity of $u, v$,
\begin{align*}
b(1)\left(\nabla H(u), v \right)&=
\mathbf{D}H(u)\cdot v =
\left\langle \frac{\delta H}{\delta u}, v \right\rangle
\stackrel{\eqref{induced_metric}} = b(1)\left(\frac{\delta H}{\delta u}, 
v \right).
\end{align*}
This shows that $\nabla H(u)= \frac{\delta H}{\delta u}$.

(ii)  For the inner product $b_1(1)$ on $\operatorname{L}(S^1)$ 
defining the normal metric, if $u,v \in \operatorname{L}(\mathbb{R})$, 
we have by periodicity of $\int_0^\theta \! {\rm d}\varphi\,  \frac{\delta H}{\delta u}
(\varphi)$ and $\int_0^\theta \!{\rm d}\varphi\, 
\left(\int_0^ \varphi  \! {\rm d}\psi\, 
\frac{\delta H}{\delta u}(\psi) \right)$,
\begin{align*}
b_1(1)(\nabla^1H(u), v) &= 
\mathbf{D}H(u)\cdot v =
\left\langle \frac{\delta H}{\delta u}, v \right\rangle = 
\frac{1}{2 \pi}\int_{-\pi}^\pi \! {\rm d}\theta \, 
 \frac{\delta H}{\delta u}(\theta) v(\theta)\\
& = \frac{1}{2 \pi}\!
\left.\left(\int_0^\theta \!{\rm d}\varphi\, 
\frac{\delta H}{\delta u}(\varphi)\right)
v(\theta) \right|_{-\pi}^\pi - \frac{1}{2 \pi}
\int_{-\pi}^\pi \! {\rm d}\theta\, 
 \left(\int_0^\theta \!{\rm d}\varphi\,  \frac{\delta H}{\delta u}(\varphi)\right)
  v'(\theta) \\
& = - \frac{1}{2 \pi}
\int_{-\pi}^\pi\! {\rm d}\theta\, 
 \frac{d}{d\theta}
 \left(\int_0^\theta \!{\rm d}\varphi \, 
 \left(\int_0^\varphi \!{\rm d}\psi \, 
\frac{\delta H}{\delta u}(\psi)\right)
\right)v'(\theta) \\
& = -\left\langle\frac{d}{d\theta}\left(\int_0^\theta \!{\rm d}\varphi\,
 \left(
\int_0^\varphi\!{\rm d}\psi \, 
 \frac{\delta H}{\delta u}(\psi)\right) \right), v' \right\rangle  
\stackrel{\eqref{normal_metric}}= 
b_1\left(-\int_0^\theta \! {\rm d}\varphi\, 
\left(\int_0^\varphi \!{\rm d}\psi \, 
\frac{\delta H}{\delta u}(\psi)\right),
v \right)
\end{align*}
which shows that $(\nabla^1H(u))(\theta)  
= -\int_0^\theta \!{\rm d}\varphi\, \left(\int_0^\varphi \! {\rm d}\psi \,  
\frac{\delta H}{\delta u}(\psi)\right)$. 

(iii) For the inner product $b_2(1)$ on $\operatorname{L}(S^1)$ 
defining the K\"ahler metric, if  $u,v \in \operatorname{L}(\mathbb{R})$, 
we have by periodicity of $\int_0^ \theta \!{\rm d}\varphi\,  \frac{\delta H}{\delta u}
(\varphi)$ and the isometry property of $\mathcal{H}$, 
\begin{align*}
b_2(1)(\nabla^2H(u), v) &= 
\mathbf{D}H(u)\cdot v =
\left\langle \frac{\delta H}{\delta u}, v \right\rangle
= \frac{1}{2\pi}\int_{-\pi}^\pi \!{\rm d}\theta\, 
\frac{\delta H}{\delta u}(\theta)
v( \theta) \\
& = \frac{1}{2 \pi}\left.\left(\int_0^ \theta \!{\rm d}\varphi \, 
\frac{\delta H}{\delta u}(\varphi)\right)v(\theta) \right|_{-\pi}^\pi - \frac{1}{2 \pi}
\int_{-\pi}^\pi\! {\rm d}\theta\, 
 \left(\int_0^ \theta \!{\rm d}\varphi\,
 \frac{\delta H}{\delta u}(\varphi) \right) v'(\theta) \\
& = -\left\langle \int_0^ \theta  \!{\rm d}\varphi\, \frac{\delta H}{\delta u}(\varphi), v' \right\rangle
= -\left\langle\mathcal{H}\int_0^\theta \!{\rm d}\varphi\, \frac{\delta H}{\delta u}
(\varphi), \mathcal{H}v' \right\rangle
\stackrel{\eqref{kahler_metric_loop}}= b_2(1)\left(-\mathcal{H}\int_0^ \theta \!{\rm d}\varphi\, 
\frac{\delta H}{\delta u} (\varphi), v\right)
\end{align*}
which shows that 
$\left(\nabla^2H(u)\right)(\theta) = -\mathcal{H}\int_0^ \theta \! {\rm d}\varphi\, 
\frac{\delta H}{\delta u} (\varphi)$.
\hfill $\blacksquare$
\end{proof}

\begin{corollary}
\label{cor_gen}
Under the same hypothesis as in Theorem \ref{general_gradients}{\rm (iii)}, 
the Hamiltonian vector field of the smooth function 
$H: \operatorname{L}(S^1) \rightarrow \mathbb{R}$ relative to the 
symplectic form $\omega$ on $\operatorname{L}(\mathbb{R})$ given by 
\eqref{cocycle} has the expression $X_{H}(u) =\int_0^ \theta \! {\rm d}\varphi\, \frac{\delta H}{\delta u} (\varphi)$
\end{corollary}

\begin{proof}
We have $X_{H}(u) = \mathcal{H}\nabla^2 H(u) 
\stackrel{{\rm (iii)}}=\int_0^ \theta \!  {\rm d}\varphi\, \frac{\delta H}{\delta u} (\varphi)$.
\hfill $\blacksquare$
\end{proof}

Of course, using Proposition \ref{passage_prop}, there are immediate 
counterparts of Theorem \ref{general_gradients} and Corollary \ref{cor_gen} 
on the loop group $\operatorname{L}(S^1)$,  which we shall not
spell out explicitly.
\medskip

The hypotheses guaranteeing the existence of the functional derivative
of $H$ relative to the weakly non-degenerate $L^2$ pairing are  quite
severe. For example, the theorem can be applied to the functional 
$H_1$ in Corollary \ref{h_one_cor},  but one needs additional smoothness.
Indeed, the first thing to check is if this functional has a functional derivative.
In fact, it does not, unless we assume that $u \in H^s(S^1)$ for
$s \geq 2$, in which case we have
\begin{align*}
\mathbf{D}H_1(u)\cdot v &=
\frac{1}{2\pi}\int_{-\pi}^\pi \! {\rm d}s\, u'(s)v'(s)  =
\left.\frac{1}{2\pi} u'(s)v(s)\right|_{-\pi}^\pi - 
\frac{1}{2\pi}\int_{-\pi}^\pi  \!   {\rm d}s \,  u''(s)v(s)
=\left\langle - u'', v \right\rangle,
\end{align*}
i.e., ${\delta H}/{\delta u} = - u''$. With this additional 
hypothesis, the gradient flow with respect to the weak inner product 
$b(1)$ defining the induced metric is given by $u_t = -u''$.

Therefore, to continue computing the other two gradients of $H_1$, 
we need to assume that $u \in H^s(S^1)$ for $s \geq 2$. Provided
this holds, to find the gradient relative to the (weak) inner product
$b_1(1)$ defining the normal metric, we have to check that both
\begin{align*}
\int_0^ \theta \! {\rm d}\varphi  \, 
\frac{\delta H}{\delta u}(\varphi) 
& = -\int_0^ \theta  \!  {\rm d}\varphi \,  u''(\varphi)  = -u'(\theta) + u'(0)\\
\int_0^\theta \!{\rm d}\varphi \, \left(\int_0^ \varphi  \! {\rm d}\psi   \, 
\frac{\delta H}{\delta u}(\psi) \right)
& =-\int_0^\theta \!   {\rm d}\varphi \, (u'(\varphi)- u'(0) ) = -u(\theta)
+ u'(0) \theta
\end{align*}
are periodic. While the first one is periodic, the second one is not unless we assume that $u'(0) =0$. With this additional hypothesis, the
gradient is given by $u_t = u$. However, we know from Corollary
\ref{h_one_cor} that neither $s\geq 2$, nor $u'(0) = 0$ is needed. In
addition, this can also be seen directly, as follows. For any $u,v \in \operatorname{L}(\mathbb{R})$, we have 
\begin{equation*}
b_1(1)(\nabla^1H(u), v) = 
\mathbf{D}H(u)\cdot v =
\frac{1}{2\pi}\int_{-\pi}^\pi  \! {\rm d}s\, u'(s)v'(s)=
\left\langle u', v' \right\rangle
\stackrel{\eqref{normal_metric}}= b_1(u,v)
\end{equation*}
which shows that $\nabla^1H(u) = u$.

The same situation occurs in the computation of the third gradient.
In the hypotheses of the theorem, we have
\[
\left(\nabla^2H(u) \right)(\theta)=
-\mathcal{H}\int_0^ \theta \! {\rm d}\varphi \, \frac{\delta H}{\delta u} (\varphi) 
= \mathcal{H}(u' - u'(0)) = \mathcal{H}u'
\]
because the Hilbert transform of a constant is zero. Thus, the gradient
flow is given in this case by 
\[
u_t=\mathcal{H}u'
\stackrel{\eqref{square_root}}=
\left(- \frac{d^2}{d\theta^2}\right)^{\frac{1}{2}}u.
\]
As before, the same result can be obtained easier and without any additional hypotheses in the following way:
\begin{equation*}
b_2(1)(\nabla^2H(u), v) = 
\mathbf{D}H(u)\cdot v =
\left\langle u', v' \right\rangle
= \left\langle \mathcal{H}u', \mathcal{H}v' \right\rangle
\stackrel{\eqref{kahler_metric_loop}}= b_2(1)(\mathcal{H}u', v).
\end{equation*}

%%%%%%%%%%%%%%%%%%%
%%%%%%%%%%%%%%%%%%%
\subsection{Symplectic structure on periodic functions}
\label{ssec:gardner}

The form of the periodic Korteweg-de Vries (KdV) equation we shall use is
\bq
u_t - 6uu_\theta + u_{\theta\theta\theta} = 0, 
\label{kdv}
\eq
where $u(t,\theta)$ is a real valued function of $t \in \mathbb{R}$ and
$\theta \in [-\pi,\pi]$, periodic in $\theta$, and $u_{\theta}:=\p u/\p \theta$.   The KdV equation is, of course, a famous integrable infinite dimensional Hamiltonian system. It is Hamiltonian
on the Poisson manifold of all periodic functions relative to the  \cite{Gardner1971}
bracket 
\begin{equation}
\label{KdV_bracket}
\{F,G\} = \frac{1}{2\pi}\int_{-\pi}^\pi\! {\rm d}\theta\, 
 \frac{\delta F}{\delta u}\frac{d}{d \theta}\frac{\delta G}{\delta u} \,, 
\end{equation}
where 
\[
F(u) = \int_{S^1}\! {\rm d}\theta\, f(u, u_\theta, u_{\theta\theta}, \ldots) 
\]
and similarly for $G$; the functional derivative 
${\delta F}/{\delta u}$ is the usual one
relative to the $L^2(S^1)$ inner product, i.e.,
\[
\frac{\delta F}{\delta u} = \frac{\partial f}{\partial u} - 
\frac{d }{d \theta}\left(\frac{\partial f}{\partial u_\theta} \right)
+ \frac{d^2}{d \theta^2}\left(\frac{\partial f}{\partial u_{\theta\theta}} \right) - \cdots .
\]
The Hamiltonian vector field of $H(u) = \frac{1}{2\pi}\int_{-\pi}^\pi \!{\rm d}\theta\, 
h(u, u_\theta, u_{\theta\theta}, \ldots) $ has the expression
\[
X_H(u) = \frac{d}{d\theta}\left(\frac{\delta H}{\delta u} \right).
\]
For the KdV equation one takes
\bq
H(u) = \frac{1}{2\pi}\int_{-\pi}^\pi \! {\rm d}\theta\,  
\left(u ^3 +  \frac{1}{2}u_\theta ^2\right).
\label{kdvham}
\eq

The Casimir functions of the Gardner bracket are all smooth functionals
$C$ for which ${\delta C}/{\delta u} = c$ is a constant function, 
i.e., 
\[
C(u) = \left\langle c, u \right\rangle = 
\frac{1}{2\pi}\int_{-\pi}^\pi \!{\rm d}\theta\,  cu(\theta)  = c \widehat{u}(0). 
\]
Thus $C^{-1}(0)$ is a candidate weak symplectic leaf in the phase space
of all periodic functions. The situation in infinite dimensions is not
as clear as in finite dimensions, where this would be a conclusion, 
because there is no general stratification theorem and one cannot expect, 
in general, more than a weak symplectic form. However, in our case, this
actually holds, as shown in \cite{ZaFa1971}. Indeed, 
\begin{align}
\label{Z_F_symplectic}
\sigma(u_1, u_2) :& = \frac{1}{4\pi}\int_{-\pi}^\pi \!{\rm d}\theta\, 
\left(\int_0^\theta\!{\rm d} \varphi\, 
\left(u_1(\varphi)u_2(\theta) - u_2(\varphi)u_1(\theta)\right)\right)
= \frac{1}{2\pi}\int_{-\pi}^\pi\!{\rm d}\theta \, 
\left(\int_0^\theta \!{\rm d}\varphi\, 
u_1(\varphi) \right) 
u_2(\theta)
\nonumber \\
&= \left\langle \int_0^\theta \!{\rm d}\varphi\,u_1(\varphi), u_2 \right\rangle
\end{align} 
defines a weak symplectic form on $\operatorname{L}(\mathbb{R})$ whose formal
Poisson bracket is \eqref{KdV_bracket}. This immediately shows that there
is a tight relationship with the symplectic form $\omega$ of the complex 
Hilbert space $\operatorname{L}(\mathbb{R})$, the Lie algebra
of the based loop groups, given by \eqref{cocycle}, namely
\[
\sigma \left( \frac{d^2}{d\theta^2}u, v\right) = 
\omega(u,v)
\]
for all $u, v \in \operatorname{L}(\mathbb{R})$ of class $H^s$, $s \geq 2$.
Defining 
\[
\left(\frac{d}{d\theta}\right)^{-1} \!\!\! u:=\int_0^\theta \!{\rm d}\varphi\, 
u(\varphi), 
\]
the KdV symplectic form $\sigma$ has the suggestive expression (see 
\eqref{der_exp})
\[
\sigma(u_1,u_2) = \left\langle\left(\frac{d}{d\theta}\right)^{-1} \!\!\! u_1, 
u_2 \right\rangle\,,
\]
which is well defined on $H^{-\frac{1}{2}}(S^1, \mathbb{R})$.

On the other hand, the Poisson bracket given by the 
K\"ahler symplectic
form \eqref{cocycle} on $\operatorname{L}(\mathbb{R})$ is
\begin{equation}
\{F,G\}=\frac{1}{2\pi}\int_{-\pi}^\pi\! {\rm d}\theta\, 
\frac{\delta F}{\delta u}
\left(\frac{d}{d\theta}\right)^{-1}\frac{\delta G}{\delta u}\,, 
\end{equation}
which is similarly well defined on $H^{-\frac{1}{2}}$,  and the Hamiltonian
vector field defined by this bracket is given by Corollary \ref{cor_gen},
i.e.,
\bq
u_t = X_H(u) = \left(\frac{d}{d\theta}\right)^{-1}
\frac{\delta H}{\delta u}\,.
\label{kdv4}
\eq
Now, the gradient vector field for the corresponding 
K\"ahler metric,
as computed in Theorem \ref{general_gradients}(iii), is 
written as
\bq
u_t = - \mathcal{H}\left(\frac{d}{d\theta}\right)^{-1}
\frac{\delta H}{\delta u}\,.
\label{kahler_flow}
\eq

%%%%%%%%%%%%%%%%%%%
%%%%%%%%%%%%%%%%%%%
%%%%%%%%%%%%%%%%%%%
\section{Metriplectic Systems} 
\label{sec_metriplectic}

In this section we define metriplectic systems and show how to construct general classes of such systems in terms of triple brackets for both finite- and infinite-dimensional theories.  We use some of the machinery developed above to address specific examples.

%%%%%%%%%%%%%%%%%%%
%%%%%%%%%%%%%%%%%%%
\subsection{Definition and consequences}
\label{metridef}

A \textit{metriplectic system} consists of a smooth manifold $P$, two 
smooth vector bundle maps $\pi,\kappa: T ^\ast P \rightarrow TP$ covering 
the identity, and two functions $H, S \in C^{\infty}(P)$, the 
\textit{Hamiltonian} or \textit{total energy} and the \textit{entropy} of 
the system, such that
\begin{itemize}
\item[(i)]\quad $\{F,G\}: = \left\langle \mathbf{d}F, \pi(\mathbf{d}G) \right\rangle$ is a Poisson bracket; in particular $\pi^\ast = - \pi$;
\item[(ii)]\quad $(F,G): = \left\langle \mathbf{d}F, \kappa(\mathbf{d}G) \right\rangle$ is a positive semidefinite symmetric bracket, i.e.,
$(\,,)$ is $\mathbb{R}$-bilinear and symmetric, so $\kappa^\ast = \kappa$,
and $(F,F) \geq 0$ for every $F \in C ^{\infty}(P)$;
\item[(iii)]\quad $\{S, F\} = 0$ and $(H, F)=0$ for all 
$F \in C ^{\infty}(P) \Longleftrightarrow \pi(\mathbf{d}S) = \kappa(\mathbf{d}H) =0$.
\end{itemize}
The \textit{metriplectic dynamics} of the system is given in terms of the two brackets by
\begin{equation}
\label{metriplectic_dynamics_function}
\frac{d}{dt}F = \{F, H+S\} + (F, H+S) = \{F, H\} + (F, S), 
\quad \text{for all}\quad F \in C ^{\infty}(P),
\end{equation}
or, equivalently, as an ordinary differential equation, by
\begin{equation}
\label{metriplectic_dynamics}
\frac{d}{dt}c(t) = \pi(c(t))\mathbf{d}H(c(t)) + 
\kappa(c(t)) \mathbf{d}S (c(t)).
\end{equation}
The Hamiltonian vector field $X_H: = \pi(\mathbf{d}H) \in 
\mathfrak{X}(P)$ represents the \textit{conservative} or 
\textit{Hamiltonian part}, whereas 
$Y_S: = \kappa(\mathbf{d}S) \in \mathfrak{X}(P)$ the 
\textit{dissipative part} of the full metriplectic dynamics 
\eqref{metriplectic_dynamics_function} or 
\eqref{metriplectic_dynamics}. 

As far as we know, first attempts to introduce such a 
structure were given in  adjacent papers by \cite{Kaufman1984} 
and \cite{Morrison84a}.  (See also \cite{KM82}.)   
\cite{Kaufman1984}  imposed, instead of (iii), the weaker 
condition $\{H,S\}=(H,S)=0$, which is enough, as will become 
apparent below, to deduce the First and Second Laws of 
Thermodynamics. In the plasma examples presented,  he used 
(iii)  for a large class of functions.   All three axioms, 
including the degeneracy condition of  (iii), were stated 
explicitly in \cite{Morrison84a} and \cite{Morrison84b}.  The 
former treated the same kinetic example as \cite{Kaufman1984} 
along with additional formalism, while the latter presented  
the metriplectic formalism for the compressible Navier-Stokes 
equations with entropy production.  All  three axioms were 
restated in \cite{Morrison1986}, where the terminology 
metriplectic was introduced and a detailed physical motivation 
for the introduction of (iii) is presented along with other 
examples such as a dissipative free rigid body equation and 
the Vlasov-Poisson equation with a collision  term that 
generalizes the Landau and Balescu-Lenard equations. In  
\cite{GrOt1997}, under the name GENERIC (General Equations for 
Non-Equilibrium Reversible Irreversible Coupling), the same 
geometric  structure was used to analyze many other equations; 
due to this paper and 
subsequent work of these authors, the metriplectic formalism 
has been popularized.  For a very interesting modern 
application of this  structure see \cite{Mielke2011} and for 
further discussion about avenues for generalization see 
\cite{Morrison2009}. 

The definition of metriplectic systems has three immediate 
important consequences. Let $c(t)$ be an integral curve of the 
system \eqref{metriplectic_dynamics}.
\begin{itemize}
\item[(1)] \quad \textit{Energy conservation}:
\begin{equation}
\label{energy_conservation}
\frac{d}{dt}H(c(t)) = \{H, H\}(c(t)) + (H, S)(c(t)) = 0.
\end{equation}
\item[(2)] \quad\textit{Entropy production}:
\begin{equation}
\label{entropy_production}
\frac{d}{dt}S(c(t)) = \{S, H\}(c(t)) + (S, S)(c(t)) \geq 0.
\end{equation}
\item[(3)] \quad\textit{Maximum entropy principle yields 
equilibria}: Suppose that there are $n$ functions
$C_1, \ldots, C_n \in C ^{\infty}(P)$ such that $\{F,C_i\} = 
(F,C_i) =0$ for all $F \in C ^{\infty}(P)$, i.e., these 
functions are simultaneously conserved by the conservative 
and dissipative part of the metriplectic
dynamics. Let $p_0 \in P$ be a maximum of the entropy $S$ 
subject to the constraints $H^{-1}(h) \cap C_1^{-1}(c_1) 
\cap \ldots C_n^{-1}(c_n)$, for given regular values 
$h, c_1, \ldots, c_n \in \mathbb{R}$ of
$H, C_1, \ldots, C_n$, respectively. By the Lagrange 
Multiplier Theorem, there exist $\alpha, \beta_1, \ldots, 
\beta_n\in \mathbb{R}$ such that
\[
\mathbf{d}S(p_0) = \alpha\mathbf{d}H(p_0) + 
\beta_1\mathbf{d}C_1(p_0) + \cdots + \mathbf{d}C_n(p_0).
\]
But then, assuming that $\alpha\neq 0$, for every 
$F \in C^{\infty}(P)$, we have
\begin{align*}
\{F,H\}(p_0) + (F,S)(p_0) &= \left\langle\mathbf{d}F(p_0), 
\pi(p_0)\left(\mathbf{d}H(p_0)\right)\right\rangle + 
\left\langle\mathbf{d}F (p_0), 
\kappa(p_0)\left(\mathbf{d}S(p_0)\right)\right\rangle\\
& = \left\langle\mathbf{d}F(p_0), \frac{1}{\alpha}
\pi(p_0)\left(
\mathbf{d}S(p_0) - \beta_1\mathbf{d}C_1(p_0)- \cdots -
\mathbf{d}C_n(p_0)\right) \right\rangle\\
&\qquad +\left\langle\mathbf{d}F (p_0), \kappa(p_0)\left(
\alpha\mathbf{d}H(p_0) + \beta_1\mathbf{d}C_1(p_0)+ \cdots + 
\mathbf{d}C_n(p_0)\right)\right\rangle\\
& = \frac{1}{\alpha}\{F,S\}(p_0) - 
\frac{\beta_1}{\alpha}\{F, C_1\}(p_0)
- \cdots - \frac{\beta_n}{\alpha}\{F, C_n\}(p_0)\\
& \qquad + \alpha(F,H)(p_0) + \beta_1(F,C_1)(p_0) + \cdots +
\beta_n(F,C_n)(p_0)=0
\end{align*}
which means that $p_0$ is an equilibrium of the metriplectic 
dynamics \eqref{metriplectic_dynamics_function} or 
\eqref{metriplectic_dynamics}.  This is akin to the free 
energy extremization of thermodynamics, as noted  by 
\cite{Morrison84b}  and \cite{Morrison1986} where it was 
suggested that one can build in degeneracies associated with 
Hamiltonian ``dynamical constraints.''  (See also
\cite{Mielke2011}.)
\end{itemize}

Suppose that $K \in C ^{\infty}(P)$ is a conserved quantity 
for the Hamiltonian part of the metriplectic dynamics, i.e., 
$\{K,H\} = 0$. Then, if $c(t)$ is an integral curve of the metriplectic dynamics, we have
\begin{align*}
\frac{d}{dt}K(c(t)) &= \mathbf{d}K(c(t))\left(\dot{c}(t)\right)
=\left\langle\mathbf{d}F(c(t)), \pi(c(t))\left(\mathbf{d}H(c(t)) \right) 
\right\rangle + \left\langle\mathbf{d}F(c(t)), \kappa(c(t))
\left(\mathbf{d}S(c(t)) \right) \right\rangle\\
& = \{K, H\}(c(t)) + (K,S)(c(t)) = (K,S)(c(t)).
\end{align*}
As pointed out in \cite{Morrison1986}, this immediately implies that
a function that is simultaneously conserved for the full metriplectic
dynamics and its Hamiltonian part, is necessarily conserved for the
dissipative part. Physically, it is advantageous for general metriplectic
systems to conserve dynamical constraints, i.e., conserved quantitates 
of its Hamiltonian part and the examples given in \cite{Kaufman1984}, \cite{Morrison84a}, \cite{Morrison84b}, and \cite{Morrison1986} satisfy this condition.

%%%%%%%%%%%%%%%%%%%
%%%%%%%%%%%%%%%%%%%
\subsection{Metriplectic systems based on Lie algebra triple brackets}
\label{metriplie}

Associated with any quadratic Lie algebra (i.e., a Lie algebra 
admitting a bilinear symmetric invariant form) is a natural completely 
antisymmetric triple bracket.  This is used to construct Lie 
algebra based metriplectic systems.  The algebra 
$\mathfrak{so}(3)$ is worked out explicitly and examples are given.

%%%%%%%%%%%%%%%%%%%
\subsubsection{General theory}
\label{generaltheory}

A quadratic Lie algebra is, by definition, a Lie algebra
admitting a bilinear symmetric non-degenerate invariant form
$\kappa:\mathfrak{g} \times \mathfrak{g}\rightarrow 
\mathbb{R}$ (the letter $\kappa$ is meant to remind one of the 
Killing form in a semisimple Lie algebra). Recall that 
invariance means that $\kappa([ \xi, \eta], \zeta) =
\kappa(\xi, [\eta, \zeta])$ for all $\xi, \eta, \zeta \in 
\mathfrak{g}$ or, equivalently, that the adjoint operators 
$\operatorname{ad}_\eta$ for all $\eta\in \mathfrak{g}$ are 
antisymmetric relative to $\kappa$. Non-degeneracy (strong) 
means that the map $\mathfrak{g} \ni \xi \mapsto
\kappa(\xi, \cdot ) \in \mathfrak{g}^\ast$ is an isomorphism.
Finite dimensional quadratic Lie algebras have been 
completely classified in \cite{MeRe1985}. For finite 
dimensional Lie algebras, non-degeneracy is equivalent to
the following statement: $\kappa(\xi, \eta) = 0$ for all 
$\eta\in \mathfrak{g}$ if and only if $\xi = 0$. In infinite 
dimensions this condition is called weak non-degeneracy
and it is implied by non-degeneracy but the converse is, in
general, false.

For example, let $\mathfrak{g}$ be an arbitrary finite 
dimensional Lie algebra. Recall that the Killing form
is defined by $\kappa(\xi, \eta) := \operatorname{Trace}
(\operatorname{ad}_\xi\circ\operatorname{ad}_\eta)$. 
If $\{e_i\}$, $i = 1, \ldots \dim \mathfrak{g}$, is an 
arbitrary basis of $\mathfrak{g}$ and $c_{\phm p ij}^p$ are 
the structure constants of $\mathfrak{g}$, i.e., $[e_i, e_j] = 
c_{\phm p ij}^pe_p $, then
\[
\kappa(\xi, \eta) = \xi^ic_{\phm p iq}^p\eta^j c_{\phm p jp}^q
\]
and hence the components of $\kappa$ in the basis $\{e_i\}$, 
$i = 1, \ldots \dim \mathfrak{g}$, are given by
\[
\kappa_{ij} = \kappa(e_i, e_j) = 
c_{\phm p iq}^p c_{\phm q  jp}^q.
\]
The Killing form is bilinear symmetric and invariant; it is
non-degenerate if and only if $\mathfrak{g}$ is semisimple.
Moreover, $-\kappa$ is a positive definite inner product if 
and only if the Lie algebra $\mathfrak{g}$ is compact (i.e.,
it is the Lie algebra of a compact Lie group).

In general, let $\kappa$ be a bilinear symmetric non-degenerate invariant form and define the completely antisymmetric covariant 3-tensor
\[
c(\xi, \eta, \zeta) : = \kappa(\xi, [\eta, \zeta])
= -c(\xi, \zeta, \eta) = - c(\eta, \xi, \zeta)
= - c(\zeta, \eta, \xi). 
\]
In the coordinates given by the basis $\{e_i\}$, 
$i = 1, \ldots \dim \mathfrak{g}$, the components of $c$
are
\[
c_{ijk}: = \kappa_{im}c_{\phm q  jk}^m = - c_{ikj} = 
- c_{jik} = - c_{kji}.
\]

This construction immediately leads to the triple bracket 
introduced by \cite{IBB(1991)} (see also \cite{Morrison1998}), 
$\{\,\cdot\,,\cdot,\cdot\}:C^{\infty}(\mathfrak{g})\times C^{\infty}(\mathfrak{g})\times C^{\infty}(\mathfrak{g})\rightarrow C^{\infty}(\mathfrak{g})$ defined by
\begin{equation}
\label{latriple}
\{f,g,h\}(\xi) : = c(\nabla f(\xi), \nabla g(\xi), \nabla h (\xi) ) : = 
\kappa\left(\nabla f(\xi),\left[\nabla g( \xi) , \nabla h (\xi) \right]\right),
\end{equation}
where the gradient is taken relative to the non-degenerate
bilinear form $\kappa$, i.e., for any $\xi\in \mathfrak{g}$
we have 
\[
\kappa(\nabla f(\xi) , \cdot ) : = \mathbf{d}f (\xi)  
\]
or, in coordinates
\[
\nabla^i f(\xi) = 
\kappa^{ij} \frac{\partial f}{\partial \xi^i}
\]
where $[\kappa^{ij}] = [\kappa_{kl}]^{-1}$, i.e., $\kappa^{ij}\kappa_{jk} = \delta^i_k$. This triple bracket is trilinear 
over $\mathbb{R}$, completely antisymmetric, and satisfies 
the Leibniz rule in any of its variables. In coordinates it
is given by
\begin{align*}
 \{f,g,h\}&=c_{ijk}\nabla^if\nabla^jg\nabla^kh
 =  \kappa_{im}c_{\phm q  jk}^m \kappa^{ip} 
 \frac{\partial f}{\partial \xi^p} \kappa^{jq} 
 \frac{\partial g}{\partial \xi^q} \kappa^{kr} 
 \frac{\partial h}{\partial \xi^r} =
c_{\phm q  jk}^p\kappa^{jq}\kappa^{kr}
\frac{\partial f}{\partial \xi^p} 
 \frac{\partial g}{\partial \xi^q} 
 \frac{\partial h}{\partial \xi^r}\\
& = c^{pqr}\frac{\partial f}{\partial \xi^p} 
 \frac{\partial g}{\partial \xi^q} 
 \frac{\partial h}{\partial \xi^r},
\end{align*}
where $c^{pqr}$ are the components of the contravariant
completely antisymmetric 3-tensor $\bar{c}$ associated to 
$c$ by raising its indices with the non-degenerate symmetric 
bilinear form $\kappa$, i.e., for any $\xi, \eta, \zeta \in \mathfrak{g}$, we have 
\[
\bar{c}\left(\kappa(\xi, \cdot),\kappa(\eta, \cdot ),\kappa(\gamma,\cdot )\right): = c(\xi, \eta, \zeta). 
\]
  
This construction extends the bracket due to \cite{nambu} to 
a Lie algebra setting.  Nambu considered ordinary vectors in 
$\mathbb{R}^3$ and defined
 \bq
 \{f,g,h\}_{\rm Nambu}(\boldsymbol{\Pi})= 
 \nabla f(\boldsymbol{\Pi})\cdot(\nabla g(\boldsymbol{\Pi})\times \nabla h(\boldsymbol{\Pi}))\,,
 \label{nambu}
 \eq
where `$\cdot$' and `$\times$' are the ordinary dot and cross products.  Thus,  the Nambu bracket is a special case of 
the triple bracket \eqref{latriple} in the case of 
$\mathfrak{g} = \mathfrak{so}(3)$,   whose the structure 
constants are the completely antisymmetric Levi-Civita symbol 
$\epsilon_{ijk}$.  Such `modified rigid body brackets'  were 
also described in  \cite{BlMa1990}, \cite{HoMa1991}, and 
\cite{MaRa1999}.
 
If $\mathfrak{g}$ is an arbitrary quadratic Lie algebra with
bilinear symmetric non-degenerate invariant form $\kappa$, 
the quadratic function 
\begin{equation}
\label{quadratic_casimir}
C_2(\xi): = \tfrac{1}{2}\kappa(\xi, \xi)
\end{equation}
is a Casimir function for the Lie-Poisson bracket on 
$\mathfrak{g}$, identified with $\mathfrak{g}^\ast$ via
$\kappa$, i.e.,
\begin{equation}
\label{LP_bracket}
\{f,g\}_\pm(\xi) = \pm \kappa\left(\xi, 
[\nabla f(\xi), \nabla g(\xi)] \right),
\end{equation}
as an easy verification shows since $\nabla C_2(\xi) = \xi$.
In view of \eqref{LP_bracket}, the following identity is obvious
\[
\{f,g\}_+ = \{C_2,f,g\}
\]
(this was first pointed out in \cite{IBB(1991)}). For example, 
if $\mathfrak{g} = \mathfrak{so}(3)$, the (-)Lie-Poisson 
bracket
\begin{equation}
\label{lpso3}
 \{f,g\}_-^{\mathfrak{so}(3)}(\boldsymbol{\Pi})
 = -\{C_2,f,g\}_{\rm Nambu}(\boldsymbol{\Pi})
 =-\boldsymbol{\Pi}\cdot\left(\nabla f(\boldsymbol{\Pi})\times \nabla g(\boldsymbol{\Pi})\right)
\end{equation}
is the rigid body bracket, i.e., if $h(\boldsymbol{\Pi}) = 
\frac{1}{2}\boldsymbol{\Pi} \cdot \boldsymbol{\Omega}$, 
where $\boldsymbol{\Pi}_i = I_i\boldsymbol{\Omega}_i$, 
$I_i>0$, $i=1,2,3$, and $I_i$ are the principal moments of 
inertia of the body, then Hamilton's equations 
$\frac{d}{dt}F(\boldsymbol{\Pi}) = 
\{f, h\}_-^{\mathfrak{so}(3)}(\boldsymbol{\Pi})$ are
equivalent to Euler's equations $\dot{\boldsymbol{\Pi}} = 
\boldsymbol{\Pi} \times \boldsymbol{\Omega}$.

Note that given any two functions, $f,g \in 
C ^{\infty}(\mathfrak{g})$, because the triple bracket
satisfies the Leibniz identity in every factor, the map
$C^\infty(\mathfrak{g}) \ni h \mapsto \{h,f,g\} \in 
C^{\infty}( \mathfrak{g})$ is a derivation and hence
defines a vector field on $\mathfrak{g}$, denoted by
$X_{f,g}: \mathfrak{g} \rightarrow \mathfrak{g}$, i.e.,
\begin{equation}
\label{ham_vf_nambu}
\left\langle \mathbf{d} h(\xi), X_{f,g}(\xi) \right\rangle = 
\kappa\left(\nabla h(\xi), X_{f,g}(\xi)\right) 
= \{h,f,g\}(\xi) \qquad \text{for all} \qquad 
h \in C^{\infty}(\mathfrak{g}).
\end{equation} 
Note that
$X_{f,f} = 0$. Thus, for triple brackets, two functions 
define a vector field, analogous to the Hamiltonian vector
field defined by a single function associated to a standard Poisson bracket.

From (\ref{latriple}) we have the following result. 
\begin{proposition}
\label{latripleprop}
The vector field $X_{f,g}$ on $\mathfrak{g}$ corresponding to the 
pair of functions $f,g$ is given by
\begin{equation}
X_{f,g}(\xi)=[\nabla f(\xi),\nabla g(\xi)]\,.
\label{latripleeqn}
\end{equation}
\end{proposition}

Triple brackets of the form \eqref{latriple} can be used to 
construct metriplectic systems on a quadratic Lie algebra 
$\mathfrak{g}$ in the following manner. Let
$\kappa$ be the bilinear symmetric non-degenerate form on 
$\mathfrak{g}$ defining the quadratic structure and fix some
$h \in C^{\infty}(\mathfrak{g})$. Define
the symmetric bracket
\begin{equation}
\label{symmetric_bracket_nambu}
(f, g)_h^\kappa(\xi): = -\kappa\left(X_{h,f}(\xi), X_{h,g}(\xi)\right).
\end{equation}
Assume that $-\kappa$ is a positive definite inner product. Then $(f,f) \geq 0$. Thus we have the manifold $\mathfrak{g}$
endowed with the Lie-Poisson bracket \eqref{LP_bracket}, 
the symmetric bracket \eqref{symmetric_bracket_nambu}, the Hamiltonian $h$, and for the entropy $S$ we take any Casimir
function of the Lie-Poisson bracket. Then the conditions 
(i)--(iii) of \S\ref{metridef} are all satisfied, because
$(h,g)_h^\kappa = -\kappa(X_{h,h}, X_{h,g}) = 
- \kappa(0,X_{h,g}) = 0$ for any $g \in 
C^{\infty}(\mathfrak{g})$. The equations of motion 
\eqref{metriplectic_dynamics_function} are in this case 
given by 
\begin{align*}
\frac{d}{dt}f(\xi) &= \kappa\left(\nabla f (\xi), \frac{d}{dt}
\xi \right)
=\{f, h\}_\pm(\xi) + (f, S)(\xi)
= \pm \kappa\left(\xi, [ \nabla f(\xi), \nabla h (\xi)]\right)
-\kappa\left(X_{h,f}(\xi), X_{h,S}(\xi)\right)\\
& = \mp \kappa\left(\nabla f(\xi), [\xi, \nabla h(\xi) ] \right) -\kappa\left(
[\nabla h(\xi),\nabla f(\xi)], [\nabla h(\xi),\nabla S(\xi)]\right)
\end{align*}
for any $f \in C^{\infty}(\mathfrak{g})$.

This gives the equations of motion
\begin{equation}
\dot{\xi}=\pm[\xi,\nabla h(\xi)]+[\nabla h(\xi),[\nabla h(\xi),
\nabla S(\xi)]]\,.
\label{tripledouble}
\end{equation}

Note that the flow corresponding to $S$ is a generalized double bracket
flow. Observe also that this flow reduces to a double
bracket flow and is tangent to an orbit of the group
if $\nabla h(\xi)=\xi$. Indeed if $h=\frac{1}{2}\kappa(\xi,\xi)$ the symmetric
bracket (\ref{symmetric_bracket_nambu}) reduces to the symmetric bracket
induced from the normal metric.

%%%%%%%%%%%%%%%%%%%
 \subsubsection{Special case of $\mathfrak{so}(3)$} 
 \label{sssec:special}

If the quadratic Lie algebra is $\mathfrak{so}(3)$, we identify it with $\mathbb{R}^3$ with the cross product as
Lie bracket via the Lie algebra isomorphism $\hat{\;}: \mathbb{R}^3\rightarrow\mathfrak{so}(3)$ given by $\hat{ \mathbf{u}} \mathbf{v} : = \mathbf{u}\times \mathbf{v}$
for all $\mathbf{u}, \mathbf{v} \in\mathbb{R}^3$. Since 
$\operatorname{Ad}_A\hat{\mathbf{u}} =\widehat{A \mathbf{u}}$,
for any $A \in SO(3)$ and $\mathbf{u}\in\mathbb{R}^3$, we conclude that the usual inner product on $\mathbb{R}^3$ is
an invariant inner product. In terms of elements of $\mathfrak{so}(3)$ we have $\mathbf{u}\cdot \mathbf{v} = - 
\frac{1}{2}\operatorname{Trace}\left(\hat{\mathbf{u}} \hat{\mathbf{v}}\right)$. We shall show below that the
metriplectic structure on $\mathbb{R}^3$ is precisely the 
one given in \cite{Morrison1986}.

Recall that the Nambu bracket is given for $\mathfrak{so}(3)$
by \eqref{nambu} and hence the symmetric bracket \eqref{symmetric_bracket_nambu} has the form

 \bqy
 \kappa(\{\Pi,h,f\}, \{\Pi,h,g\})&=& 
 \epsilon^{imn} \frac{\partial h}{\partial \Pi^m} 
\frac{\partial f}{\partial \Pi^n}\,  \delta_{ij}\, 
 \epsilon^{jst} \frac{\partial h}{\partial \Pi^s} 
\frac{\partial g}{\partial \Pi^t}
 \nonumber\\
 &=&  \epsilon^{imn} \, \epsilon_i^{{\,}st} \,  
\frac{\partial h}{\partial \Pi^m} 
\frac{\partial f}{\partial \Pi^n} 
\frac{\partial h}{\partial \Pi^s} 
\frac{\partial g}{\partial \Pi^t}
 \nonumber\\
 &=& \|\nabla h\|^2 \nabla g \cdot \nabla f - (\nabla f\cdot \nabla h) (\nabla g\cdot \nabla h)
 \label{so3sym}
\eqy
where in the third equality we have used the identity 
$\epsilon^{imn} \epsilon_i^{\,st} = \delta^{ms}\delta^{nt}- \delta^{mt}\delta^{ns}$. This coincides with \cite[equation (31)]{Morrison1986}.

With the choice $S(\boldsymbol{\Pi})=\|\boldsymbol{\Pi}\|^2/2$ and the usual rigid body Hamiltonian, the equations of motion (\ref{tripledouble}) are those for the relaxing rigid body given in \cite{Morrison1986}.

\medskip

\noindent{\bf Comments.}

\begin{itemize}
%%%%%%%%%%%%%%%%%%%
\item In three dimensions any Poisson bracket can be written as
\bq
\{f,g\}= J^{ij}\frac{\partial f}{\partial \Pi^i}
\frac{\partial g}{\partial \Pi^j} =\epsilon^{ij}_{\;\;\;k}V^k(\boldsymbol{\Pi}) \frac{\partial f}{\partial \Pi^i}\frac{\partial g}{\partial \Pi^j}
\label{3bkt}
\eq
where $i,j,k=1,2,3$, and $V\in\mathbb{R}^3$.  The last equality follows from the identification of $3\times 3$ antisymmetric matrices with vectors (the hat map discussed above).  Using the well know fact (which is easy to show directly) that brackets of the form of (\ref{3bkt}) satisfy the Jacobi identity if 
\bq
V\cdot\nabla\times V=0\,,
\label{integrab}
\eq
we conclude that 
\bq
\label{nambu_br}
\{F,G\}_f=\{f,F,G\}_{\rm Nambu}
\eq
satisfies the Jacobi identity for {\em any} smooth function $f$; i.e., unlike the general case where the theorem of   \cite{IBB(1991)} requires $f$ to be the quadratic Casimir, one obtains a good Poisson bracket for any $f$.   Thus, for the special case of three dimensions,   one can interchange the roles of Hamiltonian and entropy in the metriplectic formalism. \\

\item Thinking in terms of $\mathfrak{so}(3)^\ast$, the setting arising from reduction (see e.g.\ \cite{MaRa1999}), this 
construction leads to a natural geometric interpretation of a  metriplectic system on the
manifold $P = \mathbb{R}^3$.  With the  Poisson bracket on $\mathbb{R}^3$
of  \eqref{nambu_br},  the bundle map $\pi:
T ^\ast \mathbb{R}^3 \rightarrow T \mathbb{R}^3$ has the expression
\[
\pi_f(x, \Pi) = \left(x, \nabla f(\Pi) \times (\cdot)^\top \right)
\]
since $\mathbf{d}H(\Pi)^\top = \nabla H(\Pi)$ ($\mathbf{d}H(\Pi)$ is
a row vector and $\nabla H(\Pi)$ is its transpose, a column vector).  Now the triple bracket associated to the equation (\ref{tripledouble}) can be used to generate a  symmetric bracket given in  \cite{BlKrMaRa1994} as follows:
\bqy
\label{symmetric_r_three_bracket}
(F,G)_{BKMR}(\Pi) &=& (F,G)^{\kappa}_C=\kappa( \{\Pi,C,F\}, \{\Pi,C,G\})
\nonumber\\
&=& \left(\Pi\times \nabla F(\Pi) \right) \cdot 
\left(\Pi\times \nabla G(\Pi) \right)\,.
\eqy
where now $C=||\Pi||^2/2$.  Hence the bundle map $\kappa: T ^\ast \mathbb{R}^3 \rightarrow 
T \mathbb{R}^3$ has the expression
\[
\kappa(x, \Pi) = - \Pi\times \left(\Pi \times (\cdot)^\top \right).
\]

Thus, with the freedom to choose any quantity $S=f$ as an entropy, with the assurance that (\ref{integrab}) will be satisfied because $\nabla\times V= \nabla \times \nabla f=0$, we can take  $H=C$ and  have $\{F,S\}_f =0$ and $(F,H)=0$
for all $F \in C^{\infty}(\mathbb{R}^3)$. The equations of motion for 
this metriplectic system are 
\begin{equation}
\label{first_metriplectic_example}
\dot{\Pi} = - \Pi \times \nabla f(\Pi) - 
\Pi\times (\Pi \times \nabla f(\Pi)). 
\end{equation}
The symmetric bracket is the inner product of the two Hamiltonian vector fields on each concentric sphere. As discussed in \cite{BlKrMaRa1994},
this symmetric bracket can be defined on any compact Lie algebra by 
taking the normal metric on each coadjoint orbit. 

%%%%%%%%%%%%%%%%%%%
\item  The following set of equations were given in  \cite{Fish2005}:
\begin{equation}
\dot{\boldsymbol{\Pi}}=\nabla S(\boldsymbol{\Pi})\times 
\nabla H(\boldsymbol{\Pi})-\nabla H(\boldsymbol{\Pi})\times(\nabla H(\boldsymbol{\Pi})\times \nabla S(\boldsymbol{\Pi})).
\label{doublecrosseqn}
\end{equation}
Yet, this metriplectic system is  identical to that obtained from (\ref{tripledouble}), using (\ref{so3sym}), viz.
\bq
\dot{\boldsymbol{\Pi}} =  \{\boldsymbol{\Pi},S,H\} +\kappa\left(\{\boldsymbol{\Pi},H,\Pi\}, \{\boldsymbol{\Pi},H,S\}\right)\,, 
\eq

Replacing $H$ by $g$  in  (\ref{so3sym}) gives 
\begin{equation}
(F,G)_g((\boldsymbol{\Pi})) = \kappa\left(\{(\boldsymbol{\Pi}),g,F\}, \{(\boldsymbol{\Pi}),g,G\}\right)= 
\,(\nabla g(\Pi) \times \nabla F(\Pi)) \cdot 
 (\nabla g(\Pi) \times \nabla G(\Pi)).
\label{doublecross}
\end{equation}
Thus, the bundle map $\kappa: T ^\ast \mathbb{R}^3 \rightarrow 
T \mathbb{R}^3$ has the expression
\[
\kappa_g(x, \boldsymbol{\Pi}) = -\nabla g(\boldsymbol{\Pi})\times \left(\nabla\boldsymbol{\Pi} \times 
(\cdot)^\top \right).
\] 

{\bf Examples:} Two special cases of the equation \eqref{doublecrosseqn} are of interest.
\begin{description}
\item(i) If we take $H=\frac{1}{2}\|\boldsymbol{\boldsymbol{\Pi}}\|^2$ and 
$S=c\cdot\boldsymbol{\boldsymbol{\Pi}}$, $c$ a constant
vector, we obtain
\begin{equation}
\label{ex_1}
\dot{\boldsymbol{\Pi}}= c \times \boldsymbol{\Pi}-\boldsymbol{\Pi}\times(\boldsymbol{\Pi}\times c).
\end{equation}
\item(ii)  If we take $S=\frac{1}{2}\|\boldsymbol{\Pi}\|^2$ and $H=c\cdot\boldsymbol{\Pi}$, $c$ a constant,
we obtain
\begin{equation}
\label{ex_2}
\dot{\boldsymbol{\Pi}}=\boldsymbol{\Pi}\times c-c\times(c\times\boldsymbol{\Pi})\,.
\end{equation}
\end{description}

The equations of motion \eqref{ex_1} is an instance of double bracket
damping, where the damping is due to the normal metric, whereas
\eqref{ex_2} gives linear damping
of the sort arising in quantum systems. 
\end{itemize}

%%%%%%%%%%%%%%%%%%%
%%%%%%%%%%%%%%%%%%%
\subsection{The Toda system revisited}

%%%%%%%%%%%%%%%%%%%
\subsubsection{The  Toda lattice equation revisited}

We note that the Toda lattice equation fits into the metriplectic
picture in a degenerate but interesting fashion since it has a dual 
Hamiltonian and gradient character which may be seen 
by writing it in the double bracket form (\ref{dbflow}).
.

 It may be viewed either as the Hamiltonian part or the dissipative part of a metriplectic
system with  Hamiltonian $H=\frac{1}{2}\operatorname{Tr}L^2$ or
entropy function $S=\operatorname{Tr}LN$ respectively with the Toda lattice 
equations  in  the corresponding form  \eqref{lax} or \eqref{dbflow},
 as discussed in  Section \ref{metrics_on_orbits}. This observation 
may be extended to the Toda lattice flow on the normal form of 
any complex semisimple Lie algebra as can be see  in \cite{BlBrRa1992}.

%%%%%%%%%%%%%%%%%%%
\subsubsection{ Full Toda with dissipation}

It is possible to construct an interesting metriplectic system
which incorporates the full Toda dynamics. 

We consider the again the flow on the vector space of symmetric
matrices $\mathfrak{k}^\perp =\mathfrak{sym}(n)$ but now consider the 
flow on a generic orbit as discussed in \cite{DeLiNaTo1986} where 
it was shown that the flow is integrable. The Hamiltonian is
again $\frac{1}{2}\operatorname{Tr}L^2$
and the flow on full symmetric matrices 
is given by 
\begin{equation}
\dot{L}=[\pi_{\mathfrak{s}}L, L]
\end{equation}
with $\pi_{\mathfrak{s}}$ being the projection onto the skew symmetric
matrices in the lower triangular skew decomposition of a matrix. In this setting there are nontrivial Casimir functions of the bracket
(\ref{pb_general}). These are given as follows. For $L$ an $n\times n$ symmetric matrix set for $0\le k\le[\frac{1}{2}n]$
\begin{equation}
\text{det}(L-\lambda)_k
=\sum_{r-0}^{n-2k}E_{rk}(L)\lambda^{n-2k-r}
\end{equation}
where the subscript $k$ denotes the matrix obtained by deleting the 
first $k$ rows and the last $k$ columns.
Then $I_{1k}(L)=E_{1k}(L)/E_{0k}(L)$ are Casimir functions of the generic 
orbit in $\mathfrak{sym}(n)$ as shown in \cite{DeLiNaTo1986}.

Thus we obtain the metriplectic systems
\begin{equation}
\dot{L}=[\pi_{\mathfrak{s}}L, L]+[L,[L, \nabla I_{1k}]]
\end{equation}
where the metric is the normal metric on orbits of $\mathfrak{su}(n)$
restricted to the symmetric matrices (identified with $i$ times 
the symmetric matrices) as in \cite{BlBrRa1992}. 
Here  $H=\frac{1}{2}\operatorname{Tr}L^2$ and $S=I_{1k}$.

%%%%%%%%%%%%%%%%%%%
%%%%%%%%%%%%%%%%%%%
\subsection{Metriplectic systems for pdes:  metriplectic  brackets and examples}
\label{sec:tripledouble}

First we construct a class of metriplectic brackets based on triple brackets for infinite systems, then we consider in detail an example based on Gardner's bracket on $S^1$. Lastly,  we mention various generalizations.

%%%%%%%%%%%%%%%%%%%
\subsubsection{Symmetric brackets for pdes based on triple brackets}

Similar to \S\ref{metriplie} we can construct metriplectic flows for infinite-dimensional systems from completely antisymmetric triple brackets of the form
\bq
\left\{E,F,G\right\}=   \int_{S^1}\!{\rm d}\theta_1 \int_{S^1}\!{\rm d}\theta_2  \int_{S^1}\!{\rm d}\theta_3 \,  
\mathcal{C}_{ijk}(\theta_1,\theta_2,\theta_3)\, 
(\mathcal{P}^iE_u)(\theta_1)\, \, (\mathcal{P}^j  F_u)(\theta_2)\, \, (\mathcal{P}^kG_u)(\theta_3)\, 
\label{infinitetripleC}
\eq
where $E$, $F$, and $G$ are smooth functions on $S^1$, 
$\mathcal{C}_{ijk}$ is a smooth function on $S^1 \times S^1 \times S^1$ which is completely antisymmetric in its arguments,  so as to assure complete antisymmetry of 
$\left\{E,F,G\right\}$. In addition, we denote 
$E_u:=\delta E/\delta u$, etc.  Let $\mathcal{P}^i$, 
 $i=1,2,3$, be pseudo-differential operators.
 Evidently, the triple bracket of (\ref{infinitetripleC}) is trilinear and completely antisymmetric in $E,F,G$.

From (\ref{infinitetripleC}) and a Hamiltonian $H$, we construct a symmetric bracket as follows:
\bq
(F,G)_H=\int_{S^1}\!{\rm d}\theta'  \int_{S^1}\!{\rm d}\theta'' \left\{U(\theta'),H,F\right\}
\mathcal{G}(\theta',\theta'') \left\{U(\theta''),H, G\right\},
\label{infinitemetri}
\eq
where $U(\theta)$ in \eqref{infinitemetri} denotes the functional
\begin{equation}
U(\theta): u\mapsto \int_{S^1} {\rm d} \theta' u(\theta')\delta(\theta-\theta').
\end{equation}
We shall use this notation in subsequent expressions below. 
The `metric' $\mathcal{G}$ is assumed to be symmetric and positive semidefinite, i.e.,  the smooth function $\mathcal{G}: S^1 \times S^1 \rightarrow \mathbb{R}$ satisfies $ \mathcal{G}(\theta',\theta'')= \mathcal{G}(\theta'',\theta')$ and 
\bq
\int_{S^1}\!{\rm d}\theta'  \int_{S^1}\!{\rm d}\theta'' \, \mathcal{G}(\theta',\theta'') f(\theta')f(\theta'')\geq 0
\eq
for all functions $f \in C ^{\infty}(S^1)$.  Therefore, by construction, it is clear that (\ref{infinitemetri}) satisfies the following: 
\begin{description}
\item(i) \ \ $(F,G)_H=(G,F)_H$ for  all $F,G$, 
\item(ii) \ $(F,H)_H= 0$ for  all $F$,  and 
\item(iii) $(F,F)_H\geq 0$ for all $F$. 
\end{description}

As a special case suppose  $\mathcal{P}^i=\mathcal{P}$ for all $i=1,2,3$; then (\ref{infinitetripleC}) becomes 
\bq
\left\{E,F,G\right\}=   \int_{S^1}\!{\rm d}\theta_1 \int_{S^1}\!{\rm d}\theta_2  \int_{S^1}\!{\rm d}\theta_3 \,  
\mathcal{C}(\theta_1,\theta_2,\theta_3)\, 
\mathcal{P}(\theta_1)E_u\, \, \mathcal{P}(\theta_2)  F_u\, \, \mathcal{P}(\theta_3)G_u\,. 
\label{syminfinitetripleC}
\eq
As a further specialization, suppose  $\mathcal{C}(\theta_1,\theta_2,\theta_3)$  is given by 
\bq
\mathcal{C}(\theta_1,\theta_2,\theta_3)=A(\theta_1,\theta_2)+ A(\theta_2,\theta_3) + A(\theta_3,\theta_1)
\label{doublec}
\eq
where $A$ is any antisymmetric function, i.e., 
\bq
A(\theta_1,\theta_2)=-A(\theta_2,\theta_1)\,.
\label{aasyma}
\eq
The form (\ref{doublec}), assuming  (\ref{aasyma}), assures complete antisymmetry of $\mathcal{C}$. 

Finally, a particularly interesting, self-contained,  case would be to suppose the $A$'s come from some Poisson bracket, according to 
\begin{equation}
A(\theta_1,\theta_2) =\{U(\theta_1),U(\theta_2)\}\,.
\end{equation}
It would be quite natural to choose  the entropy,  $S$, to be a Casimir function of this bracket and to choose this bracket as the Hamiltonian part of the metriplectic system with symmetric bracket given by (\ref{infinitemetri}).   We give an example of this construction in Sec.~\ref{sssec:gardnersym}.

It is evident that one can construct a wide variety of symmetric brackets based on triple brackets.  For example, one can choose the pseudo-differential operators from the list $\{\mathcal{I}_d, d/d\theta, (d/d\theta)^{-1}, \mathcal{H}\}$, where 
$\mathcal{I}_d$ is the identity operator, and the Hamiltonian, $H$, and entropy (Casimir) $C$ could be one of the following functionals:
\bqy
H_0&=&\int_{S^1}\!{\rm d}\theta\,  u  
\label{gardcasi}
\\
H_2&=&\int_{S^1}\!{\rm d}\theta\,  u^2/2 
\label{quadkdv}
\\
H_1&=&\int_{S^1}\!{\rm d}\theta\,  u'^2/2 \\
H_{KdV}&=& \int_{S^1}\!{\rm d}\theta\,  
\left( u ^3 +   u'^2/2\right).
\eqy

In the Sec.~\ref{sssec:gardnersym}  we will construct a  metriplectic system based on  the Gardner bracket  (\ref{KdV_bracket}) of Sec.~\ref{ssec:gardner}.   To avoid complications, we choose a  simple example, yet one  that  displays  general features of a large class of  $1 + 1$  energy conserving dissipative system.

%%%%%%%%%%%%%%%%%%%
\subsubsection{Metriplectic systems based on the Gardner  bracket}
\label{sssec:gardnersym}

For simplicity we choose  $\mathcal{P}_i=\mathcal{I}_d$ for all $i$, and as mentioned above,  we suppose $A(\theta_1,\theta_2)$ is generated from the Gardner bracket \eqref{KdV_bracket}, i.e., 
\bq
A(\theta_1,\theta_2):=\{U(\theta_1),U(\theta_2)\} = \int_{S^1}\!d\theta\, \delta(\theta-\theta_1)\frac{d}{d\theta} \delta(\theta-\theta_2)
=\delta'(\theta_1-\theta_2)\,,
\label{gardner}
\eq
where prime denotes differentiation with respect to argument
and $\delta'(\theta_1-\theta_2)$ is defined by
\begin{align*}
\int_{S^1} {\rm d} \theta_1\int_{S^1} {\rm d} \theta_2
\delta'(\theta_1-\theta_2) f(\theta_1) g(\theta_2)
&=-\int_{S^1} {\rm d} \theta_1\int_{S^1} {\rm d} s\,
\delta'(s) f(\theta_1) g(\theta_1-s)
= \int_{S^1} {\rm d}\theta_1 f(\theta_1) g'(\theta_1) \\
&= - \int_{S^1} {\rm d} \theta_1\int_{S^1} {\rm d} \theta_2
\delta'(\theta_2-\theta_1) f(\theta_1) g(\theta_2)
\end{align*}
for any $f, g \in C^{\infty}(S^1)$, which shows that $\delta'(\theta_2-\theta_1) = - \delta'(\theta_1-\theta_2)$.
With this choice for $A$ we obtain
\[
\mathcal{C}(\theta_1,\theta_2,\theta_3)=\delta'(\theta_1-\theta_2)+ \delta'(\theta_2-\theta_3) + \delta'(\theta_3-\theta_1)\,,
\]
and  Eq.~(\ref{syminfinitetripleC}) becomes
\begin{align}
\label{gardtriple}
\left\{E,F,G\right\}&= \int_{S^1}\!{\rm d}\theta_1 
\int_{S^1}\!{\rm d}\theta_2  \int_{S^1}\!{\rm d}\theta_3 \,  
\left[\delta'(\theta_1-\theta_2) + \delta'(\theta_2-\theta_3) 
+ \delta'(\theta_3-\theta_1) \right]\, 
E_u(\theta_1)\,   F_u(\theta_2) \, G_u(\theta_3)\nonumber\\
& = \left(\int_{S^1} {\rm d}\bar\theta\, G_u (\bar\theta)\right)\int_{S^1} {\rm d}\theta\,F_u(\theta)E'_u(\theta) +
\left(\int_{S^1} {\rm d}\bar\theta\, E_u (\bar\theta)\right)\int_{S^1} {\rm d}\theta\,G_u(\theta)F'_u(\theta) \nonumber\\
& \qquad +
\left(\int_{S^1} {\rm d}\bar\theta\, F_u (\bar\theta)\right)\int_{S^1} {\rm d}\theta\,E_u(\theta)G'_u(\theta).
\end{align}

We shall construct a metriplectic system of the form 
\[
\dot{F}= \left\{H,F,G\right\} + 
\int_{S^1}\!{\rm d}\theta'  \int_{S^1}\!{\rm d}\theta'' 
\left\{U(\theta'),S,F\right\}
\mathcal{G}(\theta',\theta'') \left\{U(\theta''),S, G\right\},
\]
using the Gardner bracket (\ref{gardcasi}). 

Observe if we now set
$F =H_0$, the Casimir for the Gardner bracket
(\ref{gardcasi}), then, since $\delta H_0/\delta u = 1$, we
obtain 
\begin{equation}
\label{gardner_triple}
\left\{F,H_0,G\right\}= \int_{S^1}{\rm d}\theta\,  F_u G'_u
\end{equation}
which is precisely the Gardner bracket. To see this, let us
compute, for example, the integral in the third term of
\eqref{gardtriple}. Changing variables $s = \theta_3 - \theta_1$ we get
\begin{align*}
\int_{S^1}\!{\rm d}\theta_1 \int_{S^1}\!{\rm d}\theta_2  \int_{S^1}\!{\rm d}\theta_3 \,  
 \delta'(\theta_3-\theta_1)\, 
E_u(\theta_1)\, G_u(\theta_3) = 
- \int_{S^1}\!{\rm d}s \int_{S^1}\!{\rm d}\theta_3
\delta'(s)\, E_u(\theta_3-s)\, G_u(\theta_3)
= \int_{S^1} {\rm d} \theta_3\, E_u'(\theta_3)G_u\theta_3).
\end{align*} 
A similar computation shows that the first and second terms
vanish.

In order to construct the symmetric bracket in (\ref{infinitemetri}), we need the following, computed using
\eqref{gardtriple}:
\bqy
\left\{U(\theta),H, G\right\}&=&-\left( \int_{S^1}\!{\rm d}\bar\theta\, G_u(\bar\theta)\right)
\, H_u'(\theta) +
\int_{S^1}\!{\rm d}\bar\theta\,  G_u(\bar\theta) 
 H_u'(\bar\theta)
+ \left( \int_{S^1}\!{\rm d}\bar\theta\, H_u(\bar\theta)\right)   G_u'(\theta).
\label{innertrip}
\eqy

Now  with the counterpart of (\ref{innertrip}) for the functional $F$ with $U(\theta')$,  a choice for $H$, and a choice for $\mathcal{G}$, we can construct $(F,G)_H$. 
We make the following choices:
\begin{align}
& H_2(u)=
\int_{S^1}\!{\rm d} \theta \, \,  \frac{u^2}{2}\,,  \qquad \, S(u):=H_0(u)=\int_{S^1}\!{\rm d}\theta\,\,  u
\label{HS}
\\
&\mathcal{G}(\theta',\theta'')= \delta(\theta'-\theta'')\,.
\end{align} 

Now choose $H_2$ from (\ref{HS}) and insert it into (\ref{innertrip}) which gives
\bq
\left\{U(\theta),H_2, G\right\}= -\left( \int_{S^1}\!{\rm d}\bar\theta\, G_u(\bar\theta)\right)u'(\theta) +
\int_{S^1}\!{\rm d}\bar\theta\,  G_u(\bar\theta)  u'(\bar\theta) + S G_u'(\theta) 
\label{innertrip2}
\eq
and to construct the symmetric bracket \eqref{infinitemetri},
we need
\bqy
\left\{U(\theta''),H_2, S\right\}&=&  - u'(\theta'').
\eqy
Thus, the equations of motion are
\[
\frac{d}{dt}F = \{F,H_0,H_2\} + (F,S)_{H_2}
\]
where
\begin{equation}
\label{last_symmetric}
\big(F,S\big)_{H_2}=\int_{S^1}\!{\rm d}\theta'  \int_{S^1}\!{\rm d}\theta'' \left\{U(\theta'),H_2,F\right\}
\mathcal{G}(\theta',\theta'') \left\{U(\theta''),
H_2, S\right\}.
\end{equation}
This yields
\begin{equation}
\label{gd3eom}
u_t - u_{\theta} = S\,  u_{\theta\theta} + Q \qquad {\rm with}\qquad 
Q:=  \int_{S^1}\!{\rm d}\theta' |u_{\theta'}|^2\,.
\end{equation}

Equation (\ref{gd3eom}) has several interesting features. For fixed given constant $S$ and $Q$, it is a linear equation composed of  the heat equation with a source and with the inclusion of a linear advection term.  One can proceed to solve this equation by the usual method of constructing a temporal Green's function out of the heat kernel and expanding in a Fourier series.  After such a solution is constructed, one must enforce the fact that the global quantities  $S$ and $Q$  are both time dependent and, importantly, dependent on the solution so constructed.  Only after these constraints are enforced would one actually have a solution.  Pursuing this construction, although interesting,  is outside the scope of the present paper and will be treated elsewhere. 

We observe that the equation \eqref{gd3eom} is metriplectic.
Indeed, by construction, we have a Poisson bracket \eqref{gardner_triple} (the Gardner bracket) and a symmetric
bracket \eqref{last_symmetric}. Since these were constructed
out of triple brackets, property (iii) of Definition in
Section \ref{metridef} holds. Positive semidefiniteness
of the symmetric bracket follows from \eqref{innertrip2}.

The nature of the dissipation of (\ref{gd3eom}) is of particular interest in that it involves  the global quantities $S$ and $Q$.    This is reminiscent of collision operators, such as that due to Boltzmann and generalized nonlinear Fokker-Planck operators such as those due to Landau, Lenard-Balescu,  and others (see, e.g., \cite{Morrison1986}).   The usual dissipation in $1+1$ systems is local in nature (see Sec.~\ref{ssec:hybrid}) and dissipates energy.  Thus the metriplectic construction of this section has pointed to a quite natural type of dynamical system that has dynamical versions of both the first   and second laws of thermodynamics.   The pathway for constructing other systems with nonlinear and dispersive Hamiltonian components, other kinds of dissipation, etc. is now cleared, and some will be considered in  future publications.

\subsubsection{Some metriplectic generalizations}
\label{sssec:comments}

It is evident that many generalizations are possible.  We mention a few. 

\begin{itemize}
\item Without destroying the symmetries or formal metriplectic bracket properties we could allow one or both of the
functions $C$ and $\mathcal{G}$ to depend on the field variable $u$ or even contain pseudodifferential operations.  In fact, such ideas were used in similar brackets in \cite{FlierlMor} to facilitate numerical computation.
\item  It is clear how to generalize (\ref{infinitemetri}) to preserve more constraints, say $I_1,I_2,\dots$,  in addition to $H$.  One simply first constructs the completely antisymmetric multilinear brackets $\left\{E,F,G,H, \dots\right\}$ paralleling  (\ref{infinitetripleC}), and then, analogous to (\ref{infinitemetri}), constructs
\bq
(F,G)_{H,I_1,I_2,\dots}=\int_{S^1}\!{\rm d}\theta'  \int_{S^1}\!{\rm d}\theta'' \left\{U(\theta'),H,I_1,I_2,\dots, F\right\}
\mathcal{G}(\theta',\theta'') \left\{U(\theta''),H,I_1,I_2,\dots, G\right\}\,.
\label{infinitemetrimulti}
\eq
The bracket $(F,G)_{H,I_1,I_2,\dots}$ is guaranteed to be symmetric, conserve the invariants, and be positive semidefinite.  
\item It is of general interest to have metriplectic systems  of the form 
\[
\dot{F}= \left\{H,F,G\right\} + 
\int_{S^1}\!{\rm d}\theta'  \int_{S^1}\!{\rm d}\theta'' 
\left\{U(\theta'),S,F\right\}
\mathcal{G}(\theta',\theta'') \left\{U(\theta''),S, G\right\}
\]
(such as our example of Sec.~\ref{sssec:gardnersym}) for a suitably chosen function $G$; here $H$ is the Hamiltonian and $S$ is the entropy.  Exploring the mathematics of when this is possible is an area to pursue. 
\item The construction here is easily extendable to higher spatial dimensions.  For example consider the following triple bracket given in \cite{IBB(1991)}:
\bq
\{E,F,G\}=\int_{\mathcal{D}}\!d^6 z \,\,  E_f\left[F_f,G_f\right]\,,
\label{vlasovtriple}
\eq
\end{itemize}
where $z=(q,p)$ is a canonical six-dimensional phase space variable, $f(z,t)$ is a phase space density, as in Vlasov theory, the  `inner' Poisson bracket is defined by
\bq
[f,g]=f_q\cdot g_p- f_p\cdot g_q\,.
\eq
We assume that the domain $\mathcal{D}$ with boundary conditions enables us to set all surface terms obtained by  integrations by parts to zero, thereby assuring complete antisymmetry.  Inserting the quadratic Casimir $C_2:=\int_{\mathcal{D}} d^6 z \, f^2/2$ into (\ref{vlasovtriple}) gives
\[
\{F,G\}_{VP}=\{C_2,F,G\}=\int_{\mathcal{D}}\!d^6 z \, f\, \left[F_f,G_f\right]\,,
\]
 the Lie-Poisson bracket for the Vlasov-Poisson system, as given in \cite{Morrison80}.    Thus, this bracket with the quadratic  Casimir is  formally akin to the construction given in Sec.~\ref{generaltheory} (although we note it reduces to  a good bracket for any Casimir and  in this way is like the case of  $\mathfrak{so}(3)$ of Sec.~\ref{sssec:special}).
 The triple bracket of (\ref{vlasovtriple}) can be used in a generalization of the bracket of (\ref{infinitemetri}) to obtain a variety of energy conserving collision operators, with a wide choice of Casimirs as entropies.

%%%%%%%%%%%%%%%%%%%
%%%%%%%%%%%%%%%%%%%
\subsection{Hybrid dissipative structures} 
\label{ssec:hybrid}

Even if a system is not metriplectic, it is of interest to
see if it can be obtained from an equation which consists
of a Hamiltonian part and a gradient part with respect to a
suitable Poisson bracket and metric, respectively.

For KdV-like equations,  energy (the Hamiltonian) is generally not conserved when dissipation is added to the system.  This is common for physical systems, but a more complete model would conserve energy while accounting for heat loss, i.e., entropy production.  In the terminology of \cite{Morrison2009}, models that lose energy, such as those treated here and those  described by the double bracket formalism of \S\ref{double}, are \textit{incomplete}, while those that do represent dynamical models of the laws of thermodynamics, such as metriplectic systems, are termed complete. Although incomplete systems do not conserve energy, they may conserve other invariants, and building this in, represents an advantage  of various bracket formulations.   Thus, we construct incomplete hybrid Hamiltonian and dissipative dynamics by combining a Hamiltonian and a gradient vector field according to the prescription  
\bq
u_t= \{u,H\} +  (u,S)
\eq
where $u \mapsto\{u,H\} $ is a Hamiltonian vector field 
generated by $H$ and $u \mapsto (u,S)$ is a gradient vector 
field generated by $S$ (which could be $H$). Thus, $(\,,)$
is, up to a sign, an inner product on the space of 
functions $u$.

\medskip
 
Consider the following examples:

\begin{itemize}
\item  With the usual KdV Hamiltonian of (\ref{kdvham}) and the Gardner bracket of (\ref{KdV_bracket}) describing the Hamiltonian vector field, together with the choice 
\[
S(u)=H_1(u)=\frac{1}{4\pi}\int_{-\pi}^\pi \!{\rm d}\theta \, (u_\theta)^2
\]
we obtain for the gradients of Corollary \ref{h_one_cor}
\begin{description}
\item (i) $\quad u_t= \{u,H\} -  \nabla^1H_1= - u_{\theta\theta\theta}  + 6 u u_{\theta} - u $
\item (ii) $\ \ \  u_t= \{u,H\} -  \nabla H_1= - u_{\theta\theta\theta} + 6  u u_{\theta} + u_{\theta\theta} $
\item (iii) $\   \ u_t= \{u,H\} -  \nabla^2 H_1=  -u_{\theta\theta\theta} + 6  u u_{\theta} -\mathcal{H}(u_\theta)$
\end{description}
which is the KdV equation of (\ref{kdv}) with the inclusion of a new term that describes dissipation. Case (i) corresponds to simple linear damping,  case (ii) to `viscous' diffusion, and case (iii) to the equation of  \cite{OS1969} which adds a term to the KdV equation that describes  Landau damping. For these systems the KdV invariant $\int_{-\pi}^\pi \!{\rm d}\theta\, u^2$ serves as a Lyapunov function. 
\item   Choosing  $H=S=H_1$, the K\"ahler Hamiltonian flow of 
\eqref{kdv4} together with the dissipative flow generated by 
\eqref{kahler}, yields
\[
u_t= \{u,H_1\} -  \nabla^2 H_1=-  u_\theta - \mathcal{H}(u_\theta)
\]
which describes simple advection with Landau damping.  This equation possesses the damped traveling wave solution. 

\item We note that we can derive the heat equation from a symmetric bracket of the form \eqref{infinitemetri}, again with $\mathcal{G}(\theta',\theta'')=\delta(\theta'-\theta'')$.  Using this 
$\mathcal{G}$ and noting
$\{U(\theta), H_0,F\}=  G'_u(\theta)$, we obtain 
\begin{equation}
\label{Gardner_symmetric}
(F,G)_{H_0}= \int_{S^1}\!d\theta\,  F'_u G'_u\,. 
\end{equation}

Let us compute, for example, $\dot{F}(u) = (F,-H_2)_{H_0}$
(see \eqref{quadkdv}). Since $\delta H_2/ \delta u = -u$, we 
obtain
\[
\int_{S^1} {\rm d}\,\theta F_u\, \dot{u} = 
\frac{d}{dt}F(u) = (F,H_2)_{H_0} = 
-\int_{S^1} {\rm d}\theta\, F_u' u'=
\int_{S^1} {\rm d}\theta\, F_u u''.
\]
This yields 
\[
u_t =  u_{xx}
\]
which is the heat equation.

\end{itemize}
From these examples it is clear how a variety of hybrid Hamiltonian and dissipative flows can be constructed from the machinery we have developed.  For example, if we replace the KdV Hamiltonian by $
H(u)=\int_{S^1}\!{\rm d}\theta\,\left(\frac{1}{2}u\mathcal{H}(u_{\theta})+\frac{1}{3}{u^3}\right)$ we obtain the Benjamin-Ono equation with the various dissipative terms. Related ideas
apply to fluid dynamics may be found in \cite{GBHo2012}.

%%%%%%%%%%%%%%%%%%%
%%%%%%%%%%%%%%%%%%%
\begin{acknowledgement}
AMB was partially supported by NSF grants DMS-090656 and DMS-1207693.  PJM was supported by U.S. Department of Energy contract DE-FG05-80ET-53088. TSR was partially supported by Swiss NSF grant 200021-140238, and by the government grant of the Russian Federation for support of research projects implemented by leading scientists, Lomonosov Moscow State University under the agreement No.~11.G34.31.0054.
\end{acknowledgement}

%\begin{thebibliography}{300}
%
%\begin{thebibliography}{99.}

\end{document}